\documentclass[11pt]{article}

\usepackage{amsmath,amssymb,amsthm}
\usepackage{bbm}
\usepackage{framed}
\usepackage{tabu}
\usepackage{pgf}
\usepackage{tikz}
\usepackage{verbatim}
\usepackage{braket}
\usepackage{enumerate}
\usepackage{hyperref}
\usepackage{cleveref}
\usepackage{thmtools}
\usepackage{thm-restate}
\usepackage[margin=1in]{geometry}

\newcommand{\ketbra}[2]{|#1\rangle\langle#2|}

\hypersetup{
	colorlinks,
	linkcolor={blue!100!black},
	citecolor={red!100!black},
}

\newtheorem{result}{Result}
\newtheorem{theorem}{Theorem}[section]
\newtheorem{definition}[theorem]{Definition}

\newtheorem{fact}[theorem]{Fact}
\newtheorem{remark}[theorem]{Remark}
\newtheorem{lemma}[theorem]{Lemma}

\newtheorem{corollary}[theorem]{Corollary}

\newtheorem{claim}[theorem]{Claim}

\newcommand{\E}{\mathbb{E}}
\newcommand{\p}{\mathbb{P}}
\newcommand{\supp}{\mathrm{supp}}
\newcommand{\bin}{\{-1,1\}}
\newcommand{\eps}{\varepsilon}
\newcommand{\abs}[1]{\left| #1 \right|}
\newcommand{\vabs}[1]{\left\| #1 \right\|}

\newcommand{\pbra}[1]{\left( #1 \right)}
\newcommand{\sbra}[1]{\left[ #1 \right]}
\newcommand{\cbra}[1]{\left\{ #1 \right\}}

\newcommand{\opnorm}[1]{\vabs{ #1 }_\mathrm{op}}

\newcommand{\onenorm}[1]{\vabs{ #1 }_\mathrm{1}}
\renewcommand{\mid}{\,\middle\vert\,}

\newcommand{\dom}{\mathsf{dom}}

\newcommand{\poly}{\mathrm{poly}}

\def\01{\{-1,1\}}
\DeclareMathOperator{\Tr}{Tr}
\newcommand{\bhm}{\textsc{BHM}_{m,n}}
\newcommand{\forr}{\textsc{forr}}
\newcommand{\indi}{\mathbbm{1}}
\newcommand{\Rent}{\textsf{R}\|^*}
\newcommand{\Qent}{\textsf{Q}\|^*}
\newcommand{\Rpri}{\textsf{R}\|}
\newcommand{\Rpub}{\textsf{R}\|^{\mathrm{pub}}}
\newcommand{\Qpri}{\textsf{Q}\|}
\newcommand{\Qpub}{\textsf{Q}\|^{\mathrm{pub}}}
\newcommand{\Roneent}{\textsf{R1}^*}
\newcommand{\Rtwoent}{\textsf{R2}^*}
\newcommand{\Rone}{\textsf{R1}}
\newcommand{\Rtwo}{\textsf{R2}}

\newcommand{\SMP}{\textsf{SMP}}
\newcommand{\Had}{H_n}
\newcommand{\pmset}[1]{\{-1,1\}^{#1}} 

\newcommand{\Nbb}{\mathbb{N}}
\newcommand{\Rbb}{\mathbb{R}}
\newcommand{\X}{\mathcal{X}}
\newcommand{\Y}{\mathcal{Y}}
\newcommand{\Mcal}{\mathcal{M}}
\newcommand{\Hcal}{\mathcal{H}}

\newcommand{\Ucal}{\mathcal{U}}
\newcommand{\Pcal}{\mathcal{P}}

\newcommand{\Dcal}{\mathcal{D}}
\newcommand{\Ccal}{\mathcal{C}}
\newcommand{\Scal}{\mathcal{S}}
\newcommand{\Tcal}{\mathcal{T}}
\newcommand{\Ycal}{\mathcal{Y}}
\newcommand{\Ncal}{\mathcal{N}}
\newcommand{\equivalent}{locally equivalent~}
\newcommand{\equbits}{qubits of entanglement}

 \author{
	Srinivasan Arunachalam\\[2mm]
	IBM Quantum, Almaden Research Center\\
	\small \texttt{Srinivasan.Arunachalam@ibm.com}
	\and
	Uma Girish\\[2mm] 
	Princeton University\\
	\small \texttt{ugirish@cs.princeton.edu}
}
\date{}
\begin{document}
	\title{Trade-offs between Entanglement and Communication}
	\maketitle

\begin{abstract}
	
	
	We study the advantages of quantum communication models over classical communication models that are equipped with a limited number of \equbits. In this direction, we give explicit partial functions on $n$ bits for which reducing the entanglement increases the classical communication complexity exponentially. Our separations are as follows. For every~$k\geq~1$:

 	\vspace{1mm}
	
	\textbf{$\Qent$ versus $\Rtwoent$:} We show that quantum simultaneous protocols with $\tilde{\Theta}(k^5\log^{3} n)$ \equbits~can exponentially outperform two-way randomized protocols with $O(k)$ \equbits. This resolves an open problem from~\cite{DBLP:journals/qic/Gavinsky08} and improves the state-of-the-art separations between quantum simultaneous protocols with entanglement and two-way randomized protocols without entanglement~\cite{gavinsky2019quantum,girish2022quantum}.
	
	\vspace{1mm}
	
	
	\textbf{$\Rent$ versus $\Qent$:} We show that classical simultaneous protocols with $\tilde{\Theta}(k\log  n)$ \equbits~can exponentially outperform quantum simultaneous protocols with $O(k)$ \equbits, resolving an open question from~\cite{gavinsky2006bounded,gavinsky2019quantum}. The best result prior to our work was a relational separation against protocols without entanglement~\cite{gavinsky2006bounded}.

	\vspace{1mm}
 
	\textbf{$\Rent$ versus $\Roneent$:} We show that classical simultaneous protocols with $\tilde{\Theta}(k\log  n)$ \equbits~can exponentially outperform randomized one-way protocols with $O(k)$ \equbits.  Prior to our work, only a relational separation was known~\cite{DBLP:journals/qic/Gavinsky08}.

	
	
	Our techniques can also be used to show advantages of quantum communication models over {hybrid classical-quantum} models, i.e., models that have a large amount of both classical communication and quantum simultaneous communication. 
\end{abstract}

\section{Introduction}
One of the central goals in complexity theory is to understand the power of different computational resources. In the past four decades, communication complexity has provided a successful toolbox to establish various results in different areas of research in theoretical computer science such as circuit complexity~\cite{KarchmerW90,karchmer1995super}, streaming algorithms~\cite{kapralov2014streaming}, property testing~\cite{blais2012property}, extension complexity~\cite{fiorini2015exponential}, data structures~\cite{miltersen1995data}, proof complexity~\cite{huynh2012virtue}. 
In the standard two-player model of communication complexity introduced by Yao~\cite{yao1979some} there are two parties Alice and Bob whose goal is to compute a partial function $F:\mathcal{X}\times \mathcal{Y}\rightarrow \{-1,1,\star\}$. Alice receives $x\in \mathcal{X}$ (unknown to Bob) and Bob receives $y\in \mathcal{Y}$ (unknown to Alice) and their goal is to compute $F(x,y)$ for all $(x,y)\in F^{-1}(1)\cup F^{-1}(-1)$, while minimizing the amount of communication.  In this setting, there are three models of communication in increasing order of strength:
\begin{enumerate}[$(i)$]
	\item Simultaneous message passing ($\SMP$) model: Alice and Bob send a message to a referee Charlie, whose goal is to output $F(x,y)$.
	\item One-way model: Alice sends a message to Bob, whose goal is to output $F(x,y)$.
	\item  Two-way model: Alice and Bob can exchange several rounds of messages and their goal is to output $F(x,y)$.
\end{enumerate}
In all these models, the complexity of the protocol is the total number of bits used to describe the message. It is not hard to see that the communication complexity in model $(i)$ is at least the complexity in model $(ii)$ which in turn is at least the complexity in model $(iii)$. 

One variant of these models is when the players are allowed to use \emph{quantum resources}, for instance, the players could  send quantum messages or share entanglement. Over the past two decades, several works have established the advantage of quantum over classical communication complexity in various settings. In a sequence of works~\cite{buhrman1998quantum,buhrman2001quantum,raz1999exponential,gavinsky2006bounded,gavinsky2007exponential,klartagregev,gavinsky2020bare}, it has been shown that quantum communication can exponentially outperform classical communication.  
In particular, a few works~\cite{gavinsky09,gavinsky2019quantum, girish2022quantum} have demonstrated communication tasks that are easy to solve in the $\SMP$ model if the players share entanglement, however, every interactive~randomized~protocol without entanglement has exponentially larger cost. This leads to a natural and fundamental question (which has been asked many times before~\cite{jain2007direct,DBLP:conf/coco/CoudronH19,shi2005tensor,DBLP:journals/qic/Gavinsky08}): \emph{How much entanglement do quantum protocols really need?} Given any small-cost quantum protocol, can we simulate it by a small-cost quantum protocol that uses only a {small amount of entanglement}? Answering this question is one of the central questions in quantum communication complexity; in fact giving \emph{any}  upper bound on the number of qubits in a potentially helpful shared state has been open for~decades.


A similar  question of how much shared \emph{randomness} is necessary in classical communication complexity  is well understood. In a famous result,  Newman~\cite{newman1991private} showed that to solve communication tasks on $n$-bit inputs, with an additive overhead of $O(\log n)$ bits in communication one can assume that the players  only  have  private~randomness. 
Jain et al.~\cite{jain2005prior} showed that blackbox arguments similar to the one in~\cite{newman1991private} cannot be used to reduce the entanglement in a quantum protocol. Motivated by the question of how much entanglement protocols need, we study a fine-grained variant of this question, which will be the topic of this~work.

\begin{quote}
	\emph{Can we reduce the entanglement in a quantum communication protocol from $k$ qubits to $k/\log n$ qubits using a classical protocol of only polynomially larger cost?}
\end{quote}
In this direction, Shi~\cite{shi2005tensor} showed that we can remove any amount of entanglement using a classical communication protocol of exponentially larger cost. Subsequently,~\cite{DBLP:journals/qic/Gavinsky08,jain2007direct}  showed that this exponential blowup is inevitable, in particular they constructed a \emph{relational} problem for which we cannot reduce the entanglement with just a polynomial overhead using one-way communication alone.  Their works  left open the question of reducing entanglement in a quantum protocol computing a \emph{partial} function, using two-way classical communication between the players.\footnote{Relational separations are known as the ``weakest" form of separations between communication models. A partial function separation immediately implies a relational separation, however, the converse is false~\cite{GKW06}.}

\subsection{Main Result}
In this work, we provide a strong negative answer to this question. We give partial functions for which, reducing the entanglement by even a logarithmic factor, increases the communication cost by an exponential factor. To discuss our results, we set up some notation first. Let $\Rent$ (resp.~$\Qent$) denote the $\SMP$ communication model where Alice and Bob share entanglement and send classical (resp.~quantum) messages to the referee. Let $\Roneent,\Rtwoent$ be the one-way and two-way models of classical communication where Alice and Bob share entanglement.
The models $\Rone$ and $\Rtwo$ are similarly defined with the difference being that Alice and Bob don't share entanglement. The model $\Qpub$ is also defined similarly to $\Qent$ but without entanglement, additionally, the players are allowed public randomness. 
We first summarize our results informally~below. All these results hold for every $k\ge 1$ which is any parameter that is allowed to depend on $n$. 

Our {first result} shows that for simultaneous quantum protocols, more entanglement cannot be simulated by two-way classical communication with less entanglement (and a polynomial overhead).

\begin{result}
	\label{result1}
	There is a  partial function on $\tilde{O}(kn)$ bits that can be computed in $\Qent$ with $\tilde{O}(k^5\log^{3} n)$ qubits of communication and entanglement, but if the players only share $O(k)$ \equbits,  requires  ${\Omega}(n^{1/4})$ bits of communication in the $\textsf{R2}\hspace{0.5mm}^*$ model. 
\end{result}

There are two ways to view this result: $(i)$ It shows that in the rather weak quantum $\SMP$ model, reducing the entanglement by a polylogarithmic factor increases the classical communication by an exponential factor, even if Alice and Bob are allowed to interact.  This answers an open question in~\cite{DBLP:journals/qic/Gavinsky08}.
$(ii)$ This result can also be viewed in the context of quantum versus classical separations in communication complexity. As we mentioned earlier, numerous works~\cite{buhrman1998quantum,raz1999exponential,gavinsky2007exponential,klartagregev,gavinsky2020bare} have shown that quantum provides exponential savings for partial functions in various settings. The state-of-the-art separations between quantum and classical communication complexity for partial functions are due to~\cite{gavinsky2019quantum,girish2022quantum}; they show separations between $\Qent$ and $\Rtwo$. One drawback of the aforementioned works, in the context of our work, is that the lower bound can only be made to work for protocols where Alice and Bob share $\ll \log n$ \equbits. We improve upon this by showing separations between $\Qent$ (with more entanglement) and $\Rtwoent$ (with less entanglement). Our result can thus be seen as the current best-known separation between quantum and classical communication complexity for partial functions. 
In particular, we give a lower bound technique against $\Rtwoent$ protocols with $O(\log^c n)$ \equbits~for every $c\in\Nbb$. To the best of our knowledge, there were no known lower bound techniques that distinguished $\Rtwoent$ (with more entanglement) and $\Rtwoent$ (with less entanglement) once the number of \equbits~is $\gg \log n$, even for \emph{relational problems}. 




Our second result shows that for $\SMP$ protocols where the players share entanglement but only send classical messages, entanglement cannot be reduced even by quantum simultaneous protocols or by one-way classical protocols (with a polynomial~overhead).



\begin{result}
	\label{result2}
	There is a partial function on $\tilde{O}(kn)$ bits that can be computed in~$\Rent$  using $\tilde{O}(k\log n)$~bits of communication and $\tilde{O}(k \log n)$ \equbits, but if the players share $O(k)$ \equbits, requires ${\Omega}(n^{1/3})$ qubits of communication in the $\Qent$ model and ${\Omega}(\sqrt{n})$ bits in the~$\Roneent$~model.
\end{result}

We remark that the trade-offs obtained in this result are more fine-grained in comparison to Result~\ref{result1}, i.e.,  our separations hold even if we reduce the entanglement by a $O(\log n)$-factor. Prior to our work, the best known separation between $\Rent$ (with more entanglement) and $\Qent$ (with less entanglement) was a relational separation between $\Rent$ and $\Qpub$~\cite{gavinsky2006bounded}. Their work left open two questions: $(i)$ Does there exist a \emph{partial function} separating $\Rent$ and $\Qpub$? The weaker question of showing a functional separation between $\Qent$ and $\Qpub$ was also open and recently asked by~\cite{gavinsky2019quantum}. $(ii)$ Is there a \emph{relational} separation between $\Qent$ (with more entanglement) and $\Qent$ (with less entanglement)? Our result answers both these questions. Firstly, we prove separations for partial functions improving upon the relational separations;  secondly, we also show lower bounds for $\Qent$ with limited entanglement.  With regards to separations between $\Rent$ (with more entanglement) and $\Roneent$ (with less entanglement), prior to our work these were established in~\cite{DBLP:journals/qic/Gavinsky08,jain2007direct}, again for relational problems. Gavinsky~\cite{DBLP:journals/qic/Gavinsky08} left open the question of showing a similar separation for \emph{partial} functions and our work resolves this.

In the next two sections, we discuss the problems witnessing these separations followed by the proof sketches. Our first result is based on the Forrelation problem and the second result is based on the Boolean Hidden Matching problem.

\subsection{Result 1: Separations based on the Forrelation problem}

\subsubsection{Problem Definition: The Forrelation Problem} 
The Forrelation problem was first introduced by Aaronson in the context of query complexity~\cite{aaronson2010bqp} and subsequently has been studied again in the context of separating quantum and classical computation~\cite{raz2022oracle,aaronson2015forrelation}. Variants of the Forrelation problem have been used to show various quantum versus classical separations in communication complexity~\cite{girish2022quantum,bansal2021k,SherstovSW21,DBLP:conf/approx/GirishRZ21}. The state-of-the-art separations for quantum versus classical communication complexity of partial functions are between $\Qent$ and $\Rtwo$; one such separation is due to \cite{girish2022quantum} and is based on the Forrelation problem, which we define now.

\begin{definition}[Forrelation Function] \label{def:forrelation_function}
	Let $n\in \Nbb,n\ge 2$ be a power of two. Let $\Had$ be the (unitary) $n\times n$ Hadamard matrix. For $z_1,z_2\in \bin^{n/2}$, define the \emph{forrelation} function as
	$$
	\mathrm{forr}(z_1,z_2)=\frac{1}{n}\langle  z_2, \Had (z_1)\rangle. 
	$$
\end{definition}
Let $\eps\in(0,1]$ be a parameter. We typically set $\eps=\Theta\pbra{\tfrac{1}{\log n}}$ if it is not specified. We are interested in the communication complexity version of the Forrelation problem defined below.  

\begin{definition}[The Forrelation Problem]
	In the Forrelation problem,
	Alice is given $x\in \bin^{n}$, Bob is given $y\in \bin^{n}$. Their goal is to compute $\forr(x,y)$ given by
	\[
	\forr(x,y)= \begin{cases} 
		-1 & \mathrm{forr}(x\odot y)\geq \varepsilon/4 \\
		1 & \mathrm{forr}(x\odot y)\leq \varepsilon/8. 
	\end{cases}
	\]\label{definition:forrelation_problem}
\end{definition}
Here, $\odot$ denotes the pointwise product. Let $k\in\Nbb$ be a parameter satisfying $k=o(n^{1/50})$. We are interested in the XOR of $k$ copies of the Forrelation problem. This problem was first studied in~\cite{DBLP:conf/approx/GirishRZ21} in the context of XOR lemmas.

\begin{definition}[$\oplus^k$-Forrelation Problem]
	\label{def:XORforrelation}
	This problem is the XOR of $k$ independent instances of the Forrelation problem where $\eps=\frac{1}{60k^2 \ln n}$. To be precise, Alice and Bob receive $x=(x^{(1)},\ldots,x^{(k)})$ and $y=(y^{(1)},\ldots,y^{(n)})$ where $x^{(i)},y^{(i)}\in\bin^{n}$ for all $i\in[k]$, and they need to compute
	\[\forr^{(\oplus k)} (x,y)=\prod_{i=1}^k \forr\pbra{x^{(i)},y^{(i)}}. \]
\end{definition}

\subsubsection{Main Theorem} We now state our main theorem. For $n\in \Nbb$, let  $k\in\Nbb$ be a parameter satisfying $k=o(n^{1/50})$.
\begin{restatable}[]{thm}{theoremRtwo}
	\label{theoremRtwo}
	The $\oplus^k$-Forrelation problem can be solved with $\tilde{O}(k^5\log^3 n)$ qubits of communication in the $\Qent$ model if Alice and Bob share $\tilde{\Theta}(k^5\log^3 n)$ EPR pairs. However, if they share  $O(k)$ \equbits, then this problem requires  ${\Omega}(n^{1/4})$ bits of communication even in the~$\Rtwoent$~model.
\end{restatable}



We make a few remarks. First, the upper bound holds provided Alice and Bob share $\tilde{\Theta}(k^5 \log^3 n)$ EPR pairs, however, the lower bound holds for all possible entangled states on $O(k)$ qubits, not necessarily EPR pairs. See \Cref{sec:communication_model} for a formal description of the models and EPR pairs. Second,  although~\Cref{theoremRtwo} is stated for bounded-error models, our lower bound also holds for protocols with advantage $2^{-o(k)}$. 
To prove our lower bound, our main technical contribution is to show a \emph{Fourier growth bound for $\Rtwoent$ protocols} with limited entanglement.\footnote{For the introduction, we will loosely say ``Fourier growth of communication protocols'', when strictly speaking, we are referring to Fourier growth of \emph{XOR-fibers} and other functions associated with communication protocols.} In the following lemma, $O_\ell(t)$ is a shorthand notation for $O(t\cdot 2^{O(\ell)})$.

\begin{restatable}{lem}{lemmaRtwo}
	\label{lemma:main_lemma}
	Let $C:\bin^n\times \bin^n\to [-1,1]$ be an $\Rtwoent$ protocol of cost $c$ where Alice and Bob share an entangled state on at most $2d$ qubits for some parameter $d\in \Nbb$. Let
	$H$ be the XOR-fiber of $C$ as in~\Cref{definition:XOR_fiber}. Then, for all $\ell\in \Nbb,$ we have
	\[L_{1,\ell}(H)\triangleq \sum_{|S|=\ell}\abs{\widehat{H}(S)}\le 2^{5d}\cdot O_\ell(c^\ell). \]
\end{restatable}
We also study lower bounds for the Forrelation problem in the quantum SMP model. Prior to our work, there was no partial function separating $\Qent$ with more entanglement and $\Qent$ with less entanglement. In particular, it was unknown whether the Forrelation problem can be solved in the $\Qpub$ model with small cost and no entanglement. Our result shows that this is not the case, and that the Forrelation problem separates $\Qent$ from $\Qpub$, resolving an open problem from~\cite{gavinsky2019quantum}. 

\begin{restatable}[]{thm}{theoremQparallel}
	\label{theoremQparallel}
	The Forrelation problem requires ${\Omega}(n^{1/4})$ qubits of communication in the $\Qpub$~model.
\end{restatable}

This is also proved using a Fourier growth bound.

\begin{restatable}{lem}{lemmaQpri}\label{lemma:main_lemma_2}
	Let $C:\bin^n\times \bin^n\to [-1,1]$ be a $\Qpub$ protocol of cost $c$ and let
	$H$ be its XOR-fiber as in~\Cref{definition:XOR_fiber}. Then, for all $\ell\in \Nbb$,~we~have
	\[
	L_{1,\ell}(H)\le  O_\ell\big(c^\ell\big). 
	\]
\end{restatable}
We remark that our techniques also implies that for the $\oplus^k$-Forrelation problem, if Alice and Bob only share $O(k)$ EPR pairs, they require ${\Omega}(n^{1/4})$ \emph{qubits} of communication in the $\Qent$ model. We don't show the details of this proof, instead, we prove a stronger result, namely Theorem~\ref{theoremRoneQsmp}. 
We give an example of a partial function such that with $\Theta(k\log n)$ EPR pairs, it is solvable in the $\Rent$ model with cost $O(k\log n)$, however, with only $O(k)$ \equbits~requires cost ${\Omega}(n^{1/3})$ even in the $\Qent$ model. 

\subsection{Result 2: Separations based on Boolean Hidden Matching}
\subsubsection{Problem Definition: The Boolean Hidden Matching Problem}  We first define the Boolean Hidden Hatching problem.  The (relational) Hidden Matching problem was first defined by Bar Yossef et al.~\cite{bar2008exponential}. The Boolean Hidden Matching problem was defined by Gavinsky et al.~\cite{gavinsky2007exponential} in the context of one-way communication complexity and was used to separate the $\Rent$ and $\Rone$ models.
Subsequently this problem and its variants have found several applications,  especially in proving streaming lower bounds starting with the seminal work of Kapralov et al.~\cite{kapralov2014streaming}. The Boolean Hidden Matching problem, denoted $\bhm$, is defined as follows. Let $n,m\in \Nbb$ be parameters and $m=\alpha n$ for a small enough constant $\alpha\ll 1$.

\begin{definition}[Boolean Hidden Matching]
	\label{def:booleanmatching}    
	Alice gets $x\in \{-1,1\}^n$, Bob gets a matching on $[n]$ with $m$ edges and a string $y\in\{-1,1\}^m$. Their goal is to compute $\bhm(x,y,M)$
	given by   
	\[\bhm(x,y,M)=\begin{cases} -1 & \text{ if }Mx=\overline{y}\\
		1&\text{ if }Mx=y.\end{cases}
	\]
	%
	%
\end{definition}

Here, we use $M x\in \bin^m$ to denote the vector whose $k$-th coordinate is $x_{i_k}\cdot x_{j_k}$ for $k\in [m]$, where the edges of $M$ are $(i_1,j_1),\ldots,(i_m,j_m)\in[n]^2$. We also use $\overline{y}$ to denote $-y$.  Below we will be concerned with computing the XOR of $k$ independent copies of $\bhm$. 
\begin{definition}[$\oplus^k$-Boolean Hidden Matching Problem] \label{definition:XORBHM} This problem is the XOR of $k$ independent instances of the Boolean Hidden Matching problem. To be precise, Alice receives $x=(x^{(1)},\ldots,x^{(k)})$ and Bob receives $y=(y^{(1)},\ldots,y^{(k)})$ and $M_1,\ldots,M_k$ where $x^{(i)}\in\bin^{n}$, $y^{(i)}\in \bin^m$ and $M_i$ is a matching on $[n]$ with $m$ edges for all $i\in[k]$. They need to compute
	\[
	\bhm^{(\oplus k)}(x,y)=\prod_{i=1}^k \bhm\pbra{x^{(i)},y^{(i)},M_i}. 
	\]
\end{definition}



\subsubsection{Main Theorem} We now state our main theorem. Here, $\alpha\ll 1$ is some absolute constant and $k\in \Nbb$ is a parameter, possibly depending on $n\in \Nbb$.

\begin{restatable}[]{thm}{theoremRoneQsmp}
	\label{theoremRoneQsmp}
	The $\oplus^k$-Boolean Hidden Matching problem  can be solved with $\tilde{O}(k\log n)$ bits of communication in the $\Rent$  model if Alice and Bob share $\tilde{\Theta}(k\log n)$ EPR pairs. However, if Alice and Bob only share $O(k)$ \equbits, then this problem requires
	\begin{itemize}
		\item$ {\Omega}(k \sqrt{n})$ bits of communication in the $\Roneent$ model,
		\item ${\Omega}(k n^{1/3})$ qubits of communication in the $\Qent$ model.
	\end{itemize}
\end{restatable}


Similar to the discussion below~\Cref{theoremRtwo}, 
our upper bound holds provided Alice and Bob share $\tilde{\Theta}(k \log n)$ EPR pairs, however, the lower bound holds for all possible entangled states on $O(k)$ qubits,  and  our lower bound also holds for protocols even with advantage $2^{-o(k)}$. The main technical contribution of this part is to argue that $\Rone$ and $\Qpub$ protocols satisfy an XOR lemma with respect to computing the Boolean Hidden Matching problem. 

\paragraph*{XOR Lemmas.} XOR lemmas study the relation between the computational resources of $F$ and the $k$-fold XOR of  $F$ on $k$ independent inputs. In particular, XOR lemmas for communication complexity are of the following format: If cost-$t$ protocols have advantage at most $2/3$ in computing $F$, then cost-$o(tk)$ protocols have advantage at most $2^{-\Theta(k)}$ in computing the $k$-fold XOR of $F$. XOR lemmas provide a framework to construct hard objects in a black-box~way and have applications to several areas in theoretical computer science such as one-way functions, pseudorandom generators and streaming algorithms. We prove an XOR lemma for $\Rone$ and $\Qpub$ protocols with respect to computing the Boolean Hidden Matching problem.

\begin{restatable}[]{lem}{lemmaQpriXOR}\label{lemma:main_lemma_3}
	Let $C$ be any $\Qpub$ protocol of cost $c$. Then its advantage in computing the $\oplus^k$-Boolean Hidden Matching problem is at most 
	$
	O_k \pbra{\frac{(c/k)^{3}}{n}}^{k/2} + O_k(n^{-k/2}).
	$
\end{restatable}



\begin{restatable}[]{lem}{lemmaRoneXOR}\label{lemma:main_lemma_4}
	Let $C$ be any $\Rone$ protocol of cost $c$. Then its advantage in computing the $\oplus^k$-Boolean Hidden Matching problem is at most $O_k\pbra{\tfrac{(c/k)^2}{n}}^{k/2} + O_k(n^{-k/2})$.
\end{restatable} 

Until very recently~\cite{DBLP:conf/focs/Yu22}, we didn't have an XOR lemma for $\Rone$ and as far as we are aware we do not have \emph{any} XOR lemmas for the quantum communication model. However we do have direct product and direct sum theorems for  classical and quantum communication models (which are strictly weaker than XOR lemmas) and in fact this was used in the prior work of~\cite{gavinsky2006bounded,jain2007direct}.  Our main technical contribution here is an XOR lemma for $\Rone$ and $\Qpub$  protocols for the Boolean Hidden Matching problem. Only during completion of this project, we were made aware of a recent work by Yu~\cite{DBLP:conf/focs/Yu22} proving an XOR lemma for all constant round classical protocols (hence implying~\Cref{lemma:main_lemma_4}). Given the technicality of his proof, in our paper we present a simple proof for an XOR lemma for the $\Rone$ model for the Boolean Hidden Matching problem.


\subsubsection{Further Results: Hybrid Classical-Quantum Models}
Our techniques can be used to show lower bounds for \emph{hybrid classical-quantum} models that have a large amount of classical \emph{and} quantum communication.  We measure the communication cost by the sum of the quantum and classical cost. As far as we know, ours is the first work  proving quantum advantages over hybrid classical-quantum models. 

We first consider the $(\Rtwo+\Qpub)$ model where Alice and Bob engage in classical interactive communication, and then each send a quantum message to Charlie, who applies an arbitrary projective measurement and returns the outcome as the answer. Note that the quantum state sent by the players to Charlie is allowed to depend on the classical transcript of the protocol. 
\begin{theorem}\label{theorem:hybrid_1}
The Forrelation problem requires $\Omega(n^{1/8})$ cost in the $(\Rtwo+\Qpub)$-model. 
\end{theorem}

We also consider the $(\Rone+\Qpub)$ model where Alice sends Bob a classical message and then Alice and Bob each send a quantum message to Charlie, who applies an arbitrary projective measurement and returns the outcome as the answer. Here too, the quantum state sent by the players to Charlie is allowed to depend on the classical transcript of the protocol.  

\begin{theorem}\label{theorem:hybrid_2}
The Boolean Hidden Matching problem requires $\Omega(n^{1/6})$ cost in the $(\Rone+\Qpub)$-model. 
\end{theorem}

Some remarks are in order. Firstly, we haven't optimized the lower bound parameters, and the exact bounds can likely be improved to match the parameters in our main results. Secondly, it is likely that these lower bounds can be made to hold even in the presence of limited entanglement, using techniques similar to our main results. We omit these for simplicity.
\subsection{Proof Sketch}
One of the difficulties of proving lower bounds against classical models equipped with entanglement is that these models are quite powerful; using the quantum teleportation protocol, any quantum protocol with $q$ qubits of communication can be classically simulated using $q$ EPR pairs. Thus, all known partial functions that separate quantum and classical communication complexity are easy to classically simulate in the presence of $O(\log^c n)$ EPR pairs for some small constant $c>0$. 

One approach to show a fine-grained separation between protocols with more entanglement and protocols with less entanglement is the following. Consider any communication task $F$ that exhibits an exponential separation between quantum and classical communication complexity. We have many examples of such tasks that are easy with $O(\log n)$ EPR pairs but exponentially harder in the absence of entanglement.  Consider the problem of solving $k$ independent and parallel instances of $F$. Here, the players receive $k$ pairs of inputs $(x_i,y_i)$ and need to compute $F(x_i,y_i)$ \emph{for every} $i\in[k]$. We denote this problem by $F^{(k)}$. The hope is that entanglement obeys a direct sum theorem of sorts, that is, if the players need at least $\Omega(\log n)$ \equbits~to solve the original task, then to solve $k$ independent and parallel instances, they need at least $\Omega(k\log n)$ \equbits. In particular, we might hope that protocols that compute $F^{(k)}$ using only $O(k)$ \equbits~require exponentially larger cost. There is a way to make this idea work and this was done in~\cite{DBLP:journals/qic/Gavinsky08}. We describe this idea. Assume by contradiction that we have a small-cost protocol computing $F^{(k)}$ using only $O(k)$ \equbits. 

\emph{\textbf{Step 1: Remove entanglement.}} The first step is to remove all entanglement from this protocol. To do this, we replace the entangled state on $O(k)$ qubits by the maximally mixed state on $O(k)$ qubits. Since the maximally mixed state is unentangled, the resulting protocol effectively uses no entanglement. Furthermore, the mixed state can be viewed as a probability distribution over states, where the original entangled state occurs with probability $2^{-\Theta(k)}$. It follows that this protocol succeeds with probability at least $2^{-\Theta(k)}$. 

\emph{\textbf{Step 2: Direct Product Theorems.}} The second step is to prove a direct product theorem in the absence of entanglement. Direct product theorems in communication complexity are of the following form: If for cost-$t$ protocols, the probability of solving one instance of $F$ is at most $2/3$, then for cost-$o(tk)$ protocols, the probability of solving $k$ parallel and independent instances of $F$ is at most~$2^{-\Theta(k)}$. Establishing such theorems is highly non-trivial and for one-way protocols, this was done by~\cite{DBLP:journals/qic/Gavinsky08,jain2007direct}.

Following this framework, the work of~\cite{DBLP:journals/qic/Gavinsky08} gives examples of relational problems that are easy to solve with $\Theta(k\log n)$ EPR pairs but difficult with only $O(k)$ EPR pairs. One drawback of this approach is that the task $F^{(k)}$ has many output bits, regardless of whether $F$ is a partial function or a relational problem. To get separations for functions with single-bit outputs, we need to modify this approach. We ask the players to solve the XOR of $k$ independent instances of $F$. Here, the players receive $k$ pairs of inputs $(x_i,y_i)$  and they need to compute $\prod_{i\in[k]} F(x_i,y_i)$. We denote this problem by $F^{(\oplus k)}$. We will show that there is no small-cost protocol solving $F^{(\oplus k)}$ using only $O(k)$ \equbits. To do this, we assume by contradiction that there exists such a protocol. 

\emph{\textbf{Step 1: Remove entanglement.}} We produce a small-cost protocol for $F^{(\oplus k)}$  that uses no entanglement and has success probability at least $1/2+2^{-\Theta(k)}$, i.e., the advantage is at least~$ 2^{-\Theta(k)}$.

\emph{\textbf{Step 2: XOR Lemmas.}} We establish an XOR lemma for protocols without entanglement. We show that if for cost-$t$ protocols, the probability of solving one instance is at most $2/3$, then for cost-$o(tk)$ protocols, the probability of solving the XOR of $k$ independent instances is at most $1/2+2^{-\Theta(k)}$, i.e., the advantage is at most $2^{-\Theta(k)}$.

Together, this would establish the desired result. We now discuss some of the difficulties in executing these steps and present our solutions. We first present the details of step 2 and then step 1.

\emph{\textbf{Details of Step 2.}}
One difficulty with step 2 is that XOR lemmas are stronger than direct product theorems and are thus harder to establish.  
In this work, we present XOR lemmas that are tailored for particular functions. The functions we will be interested in are the Forrelation problem and the Boolean Hidden Matching problem. For the former problem, XOR lemmas for $\Rtwo$ protocols were established in~\cite{DBLP:conf/approx/GirishRZ21}. For the Boolean Hidden Matching problem, we show an XOR lemma for the $\Rone$ and $\Qpub$ models (\Cref{lemma:main_lemma_4} and \Cref{lemma:main_lemma_3}). We now describe this part in more detail. 

Let $F$ be the Boolean Hidden Matching problem. 
One central ingredient in our XOR lemmas for F is the construction of hard distributions. The result of~\cite{gavinsky2007exponential} shows hard distributions $\Ycal$ and $\Ncal$ on the \textsc{yes} and \textsc{no} instances of $F$ respectively, such that no small-cost protocol can distinguish these distributions with $1/3$ advantage. We produce hard distributions $\mu^{(k)}_{-1}$ and $\mu^{(k)}_1$ for the $F^{(\oplus k)}$ problem such that no small-cost protocol can distinguish them with advantage $2^{-\Theta(k)}$. To get bounds of the form $2^{-\Theta(k)}$, it turns out that our distributions need to agree on moments of size at most $\Theta(k)$. Motivated by this, we define the following distributions.
\begin{equation*}
	\mu_{1}^{(k)}:= \frac{1}{2^{k-1}}\sum_{\substack{K\subseteq[k]\\|K|\text{ is even}}} \Ycal_K \Ncal_{\overline K} \quad \text{and}\quad  \mu_{-1}^{(k)}:= \frac{1}{2^{k-1}}\sum_{\substack{K\subseteq [k]\\|K|\text{ is odd}}} \Ycal_K \Ncal_{\overline K}, 
\end{equation*}
Here, $\Ycal_K \Ncal_{\overline K}$ is a product of $k$ independent distributions, where the $i$-th distribution is $\Ycal$ if $i\in K$ and is $\Ncal$. We show that the distributions $\mu_1^{(k)}$ and $\mu_{-1}^{(k)}$ are indeed distributions on the \textsc{no} and \textsc{yes} instances respectively of the $F^{(\oplus k)}$ problem, furthermore, they agree on all moments of size at most $k$. To complete the argument, we need to show that small-cost protocols cannot distinguish these distributions with more than $2^{-\Theta(k)}$ advantage. This is fairly technical and involves the use of Fourier analysis. For $\Rone$ protocols, the $k=1$ version of this was proved in~\cite{gavinsky2007exponential}. Their work in particular makes use of the level-$k$ inequality. We build on their work for $\Rone$ protocols and prove the desired XOR lemma for larger $k$ (\Cref{lemma:main_lemma_4}). For $\Qpub$ protocols, we are not aware of any works that study the communication complexity of $F$ or that analyze the Fourier spectrum of such protocols, which is a contribution in this paper. In particular we prove Fourier growth bounds for $\Qpub$ protocols~(\Cref{lemma:main_lemma_2}) as well as an XOR lemma for the Boolean Hidden Matching problem~(\Cref{lemma:main_lemma_3}). For these, we make use of a matrix version of the level-$k$ inequality~\cite{ben2008hypercontractive}. 

\emph{\textbf{Details of Step 1.}} One difficulty with step 1 is that the trick of replacing an entangled state by the maximally mixed state no longer works. It is possible for a protocol to be correct when using a particular entangled state, but wrong for every orthogonal state. In this case, executing the protocol on the maximally mixed state would bias the output towards the wrong answer. Thus, carrying out step 1 is non-trivial and in particular, difficult to do for $\Rtwoent$ protocols. We take an alternate approach for $\Rtwoent$ protocols to sidestep this difficulty, which we will describe later. We are able to carry out step 1 for $\Roneent$ and $\Qent$ protocols (\Cref{lemma:simulation_r} and \Cref{lemma:simulation_q}). Given a cost $c$ protocol for a function in the $\Roneent$ model or $\Qent$ model using at most $2d$ \equbits, we produce a cost $c+O(d)$ protocol in the $\Rone$ or $\Qpri$ model\footnote{We use $\Qpri$ to denote the private-coin version of the $\Qpub$ model.} respectively; these protocols use no entanglement and have advantage $2^{-\Theta(d)}$.  We now give an illustrative example of this simulation. 

Consider a simple $\Qent$ protocol where the entangled state consists of $d$ EPR pairs and Alice and Bob apply a unitary operator to their part of the entangled state and send all their qubits to Charlie.  If Alice's and Bob's unitary operators map $\ket{i}$ to $\ket{u_i(x)}$ and $\ket{v_i(y)}$ respectively, then the state received by the referee is the pure state $\sum_{i\in\bin^d}\ket{u_i(x)}\ket{v_i(y)}$ (ignoring the normalization). We now construct a $\Qpri$ protocol that produces the same state with probability $2^{-\Theta(d)}$, furthermore, Charlie is able to detect when this state was successfully produced. We have Alice and Bob send the pure states $\sum_{i}\ket{i}\ket{u_i(x)}$ and $\sum_{j}\ket{j}\ket{v_j(y)}$ respectively to Charlie. Charlie first projects onto states such that $i=j$ and obtains the pure state $\sum_{i\in\bin^d}\ket{i,i}\ket{u_i(x)}\ket{v_i(y)}$ with probability $2^{-d}$. He then applies Hadamard on the first $2d$ qubits and measures. He obtains the outcome $\ket{0^{2d}}$ with probability $2^{-2d}$ in which the resulting state is the pure state $\sum_{i}\ket{u_i(x)}\ket{v_i(y)}$ as in the original $\Qent$ protocol. We use similar ideas to remove entanglement from arbitrary $\Qent$ protocols. To remove entanglement from an $\Roneent$ protocol, we need to take a different approach which involves Alice sending Bob a random coordinate of a certain density matrix. We omit the details.



\emph{\textbf{Alternate Approach to Step 1 for $\Rtwoent$ Protocols.}} We now present an alternative to step~1 for $\Rtwoent$ protocols. The idea is to prove Fourier growth bounds for $\Rtwoent$ protocols. The results of~\cite{DBLP:conf/approx/GirishRZ21} imply that for protocols whose level-$2k$ Fourier growth is at most $\alpha$, their advantage in solving the $\oplus^k$-Forrelation problem is at most $\alpha\cdot n^{-k/2}+ o(n^{-k/2})$. We directly establish a Fourier growth bound on $\Rtwoent$ protocols. In particular, we show that for $\Rtwoent$ protocols of communication cost $c$ that use $d$ \equbits, their level-$\ell$ Fourier growth is at most $\poly(2^d)\cdot O_\ell(c^{\ell})$ (\Cref{lemma:main_lemma}). Choosing $\ell=2k$, $d=\Theta(k)$ and $c=\Theta(n^{1/4})$ for appropriate constants, we have that the advantage is at most  $\poly(2^d)\cdot O_\ell(c^\ell)\cdot n^{-k/2}\ll 1$. This would complete the proof. We now describe how we prove the Fourier growth bound on $\Rtwoent$ protocols (\Cref{lemma:main_lemma}).


\emph{\textbf{(1)}}  We first show that if the players share a $2d$-qubit entangled state, then we can decompose the state into a small linear combination of $\poly(2^d)$ many \emph{two}-qubit quantum states that are either unentangled, or \equivalent to $\ketbra{ \Phi^+}{\Phi^+}$, the EPR state. (By \equivalent we mean that the players can transform one state into the other using local unitaries and no communication.) This is formalized in~\Cref{lemma:decomposition}. This gives us the pre-factor of $\poly(2^d)$ in~\Cref{lemma:main_lemma}.  

\emph{\textbf{(2)}} We analyze protocols where Alice and Bob share the EPR state and bound the Fourier growth of such protocols. Observe that if they share an unentangled state, the protocol is essentially an $\Rtwo$ protocol and the work of~\cite{girish2022quantum} showed Fourier growth for such protocols. 
We strengthen this result by proving similar Fourier growth for $\Rtwoent$ protocols where Alice and Bob share the EPR state. To do this, we study the structure of such protocols. We first show that the expected output of any $\Rtwoent$ protocol of cost $c$ where Alice and Bob share the EPR state can be written~as 
$$
C(x,y)=\sum_{z\in \bin^c} \alpha_z \cdot \Tr((E_z(x)\otimes F_z(y))\cdot \rho).
$$
where $E_z(x)$ and $F_z(y)$ are positive semidefinite matrices, $\sum_{z\in \bin^c} E_z(x)\otimes F_z(y)=\mathbb{I}$, $\alpha_z\in\bin$, and  $\rho=\ketbra{ \Phi^+}{\Phi^+}\otimes \ketbra{0^{2m}}{0^{2m}}_{AB}$ for some $m\in \Nbb$ that is possibly large (\Cref{claim:structure_rtwoent}). We give some intuition on this expression. The qubits $\ketbra{0^m}{0^m}_A$ and $\ketbra{0^m}{0^m}_B$ in $\rho$ correspond to Alice and Bob's private memory respectively and $\rho$ captures the initial state of all the qubits in the system. The matrices $E_z(x)\otimes F_z(y)$ arise out of Alice's and Bob's sequence of quantum operations (i.e., POVMs) in the $\Rtwoent$ protocol and the quantity $\Tr(E_z(x)\otimes F_z(y)\cdot \rho)$ captures the probability of the transcript being $z\in \bin^c$. The number $\alpha_z\in\bin$ is 1 if and only if the transcript $z$ results in the players outputting 1. 

We now write out the Fourier expansion of the XOR fiber $H(z)=\E_{x\sim \bin^n}[C(x,x+z)]$. Using the convolution property of Fourier coefficients, we can express the Fourier coefficients of $H(z)$ in terms of the Fourier coefficients of $E_z(x),F_z(y)$. In particular, we get
\[\sum_{|S|=\ell}|\widehat{H}(S)|= \sum_{|S|=\ell}\abs{\sum_{z\in\bin^c} \alpha_z \cdot \Tr\pbra{\pbra{\widehat{E_z}(S)\otimes \widehat{F_z}(S)}\cdot  \rho}}. \]
We now use the fact that $\rho$ has exactly four non-zero entries. This zeroes out all but four coordinates of $\widehat{E}_z(S)\otimes \widehat{F}_z(S)$ in the above expression. At this point, we use the level-$k$ inequality by Lee~\cite{DBLP:conf/coco/Lee19} to bound each of the coordinates separately in terms of the entries of $E_z(x)$ and $F_z(y)$. Using the fact that $\{E_z(x)\otimes F_z(y)\}_z$ forms a POVM, we can bound the coordinates of these matrices and combine them with the bounds we obtain from the level-$k$ inequality in a nice way. Putting everything together involves some calculation and is done in~\Cref{sec:fourier_growth_rtwo}.

\paragraph{Acknowledgements.}
{We thank Vojtech Havlicek, Ran Raz and Penghui Yao for many discussions during this project. We also thank Ran Raz for feedback on the presentation.}

\section{Preliminaries}
\label{sec:prelim}

\textbf{Sets.}   For $n\in \mathbb{N}$, let $[n]=\{1,\ldots,n\}$. We use $\indi$ to denote the indicator function, i.e., for a predicate $E$, $\indi[E]$ is 1 if $E$ is satisfied and 0 otherwise. For a subset $S\subseteq[n],$ we use $\overline{S}:=[n]\setminus S$ to denote the complement of $S$. We denote the $n\times n$ identity matrix by $\mathbb{I}_{n}$, and we omit the subscript if it is implicit. 

\textbf{Big O Notation.} For simplicity in notation, for every $f,g:\Rbb_{\ge 0} \to \Rbb_{\ge 0}$ and $\ell\in \mathbb{N}$, we say that $g=O_\ell(f)$ if for some constant $c> 0$, we have $g= O\pbra{ f\cdot  2^{c\cdot \ell}}$. We say that $f=\tilde{O}(g)$ (respectively $f=\tilde{\Omega}(g)$) if for some constant $c>0$, we have $f=O(g\cdot \log^c(g))$ (respectively $f=\Omega(g\cdot \log^{-c}(g))$). We say that $f=\tilde{\Theta}(g)$ if $f=\tilde{O}(g)$ and $f=\tilde{\Omega}(g)$.

\textbf{Probability Distributions.}
Let $\Sigma$ be an alphabet and $\Dcal$ be a probability distribution over $\Sigma$. We use $x\sim \Dcal$ to denote $x$ sampled according to $\Dcal$. We use $\supp(\Dcal)$ to denote the support of the distribution $\Dcal$. We use $x\sim\Sigma$ to denote a uniformly random sample from $\Sigma$. For a function $G:\Sigma\to \Rbb^n$, we use $\E_{x\sim \Dcal}[G(x)]$ to denote the expected value of $G$ when the inputs are drawn according to $\Dcal$. Let $k\in\Nbb$, $S\subseteq[k]$ and $\Dcal,\Dcal'$ be two distributions over $\Sigma$. We use $\Dcal_S\Dcal'_{\overline{S}}$ to denote the distribution over $\Sigma^k$ which is a product of $k$ independent distributions over $\Sigma$, where the $i$th distribution is $\Dcal$ if $i\in S$ and $\Dcal'$ if $i\notin S$ for all $i\in[k]$. For distributions, $\Dcal,\Dcal'$, define the total variation distance as $\|\Dcal-\Dcal'\|_1:=\sum_i |\Dcal(i)-\Dcal'(i)|$.

\textbf{Norms.}
Let $k\in \Nbb$. For a vector $v\in\Rbb^n$, we use $\|v\|_k:=\pbra{\sum_{i\in[n]} \abs{v_i}^k}^{1/k}$ to denote the $\ell_k$-norm of $v$. For any matrix $M\in \Rbb^{n\times n}$, we use $|M|$ to denote $\sqrt{MM^\dagger}$ and we denote by $\|M\|_1$ the trace norm of $M$, that is $\|M\|_1:=\mathrm{Tr}(\sqrt{MM^\dagger})=\Tr(|M|).$ We use $\opnorm{M}:=\max_{\|v\|_2=1} (v^T M v)$ to denote the operator norm of $M$. 

\subsection{Quantum information}

\textbf{Quantum States.} Let $d\in \Nbb$ and let $\mathcal{H}$ be a Hilbert space of dimension $2^d$. This is a vector space defined by the $\mathbb{R}$-span of the orthonormal basis $\{\ket{x}:x\in\{0,1\}^d\}$. We also identify this basis with $\{\ket{i}:i\in[2^d]\}$ using the lexicographic ordering as the correspondence. We use $\ket{0^d}$ to denote the vector $\ket{0,\ldots,0}$ with $d$ zeroes. Let $\Pcal(\Hcal)$ be the set of positive semidefinite matrices in $\mathbb{R}^{2^d\times 2^d}$. Let $\mathcal{S}(\mathcal{H})$ be the set of density operators on $\mathcal{H}$, that is, matrices in $\Pcal(\Hcal)$ with trace 1. We typically use $\rho$ and $\sigma$ to refer to elements of $\mathcal{S}(\mathcal{H})$. The state of a quantum system on $d$ qubits is described by a density operator $\rho\in \mathcal{S}(\mathcal{H})$. For states that are shared between Alice and Bob, we use the subscript $A$ and $B$ on qubits to denote whether Alice or Bob own those qubits.

\textbf{Quantum State Evolution.}
Let $\Hcal,\Hcal'$ be Hilbert spaces. The evolution of a quantum state is described by a map $E:\mathcal{S}(\mathcal{H})\rightarrow \mathcal{S}(\mathcal{H'})$ which is CPTP (i.e., completely positive and trace preserving). We use the notation $E(\rho)$ to denote the image of $\rho$ under $E$. In particular, we will be interested in measurement operators. Any quantum measurement with $k$ outcomes is specified by $k$ matrices $M_1,\ldots,M_k\in\Pcal(\Hcal)$ such that $\sum_{i\in[k]} M_i^\dagger M_i= \mathbb{I}$. The probability of getting outcome $i\in[k]$ is precisely $\Tr(M_i \rho M_i^\dagger )$ and the post measurement state upon obtaining the outcome $i$ is $\frac{M_i \rho M_i^\dagger }{\Tr(M_i \rho M_i^\dagger )}$. We have a correspondence between $\{0,1\}$ and $\{1,-1\}$ defined by $0\to 1,1\to -1$, hence, we will sometimes refer to measurement outcomes ``1'' and ``0'' as ``-1'' and ``1'' respectively.

\textbf{Distance between States.} Let $\rho,\sigma\in\mathcal{S}(\mathcal{H})$ be density operators. We define the trace distance between $\rho$ and $\sigma$ to be $\frac{\|\rho-\sigma \|_1}{2}.$
We will use the following standard facts about the trace distance. Firstly, the trace distance satisfies triangle inequality. Secondly, the trace distance between $\rho$ and $\sigma$ is equal to the maximum probability with which these states can be distinguished using any single projective measurement. Thirdly, the following inequality holds as a consequence of the Von-Neumann Inequality.
\begin{fact}
	For any matrices $M,\rho\in \Rbb^{n\times n}$, we have $\Tr(M\rho)\le \opnorm{M}\cdot  \onenorm{\rho}$.\label{fact:inequality}
\end{fact}

\subsection{Communication Complexity} \label{sec:communication_model}
The goal in communication complexity is for Alice and Bob to compute a function $F:\X\times\Y\to \{-1,1,\star\}$. We interpret $-1$ as a \textsc{yes} and $1$ as a \textsc{no}. We say $F$ is a \emph{total} function if $F(x,y)\in \01$ for all $x\in \X$ and $y\in \Y$, otherwise $F$ is a \emph{partial} function.  In this paper we will be mostly concerned with partial functions and denote $\dom(F)=F^{-1}(\01)$. Here Alice receives an input $x\in\X$ (unknown to Bob) and Bob receives an input $y\in\Y$  (unknown to  Alice) promised that $(x,y)\in \dom(F)$ and their goal is to compute $F(x,y)$ with high probability, i.e., probability at least $2/3$. More formally, for any protocol $\Pcal$, we let $\mathrm{cost}(\Pcal,x,y)$ be the communication cost of $\Pcal$ when Alice and Bob are given $x,y$ as inputs. We say that $\Pcal$ computes $F$, if, for every $(x,y) \in \dom(F)$, the output of the protocol is $F(x,y)$ with probability at least $ 2/3$ (where the probability is taken over the randomness/measurements of the protocol). The communication complexity of $F$ is defined as
\begin{align*}
	\min_{\Pcal\text{ computes } F} \max_{(x,y) \in \dom(F)} \mathrm{cost}(\Pcal,x,y).
\end{align*} The messages sent are referred to as the {transcript} of the protocol. We discuss a few models of communication of interest to us.

\textbf{Simultaneous Message Passing Model.} This is a general model of communication called the \emph{simultaneous message passing} ($\SMP$) model. In this model, Alice and Bob each send a single (possibly quantum or randomized) message to a referee Charlie.
The goal is for Charlie to output $F(x,y)$ with high probability, i.e., at least $2/3$ probability.  We measure the cost of a communication protocol by the total number of bits (or qubits) received by Charlie. There are many types of simultaneous protocols.

\emph{Quantum versus Classical protocols.} We use $\Rpri$ to denote the \SMP~model where the players can only send classical messages to Charlie. We use $\Qpri$ to denote the \SMP~model where the players can send a quantum message to Charlie.

\emph{Public-coin versus Private-coin Protocols.} If we allow the players to use public coins, we use the superscript ``$\mathrm{pub}$''. For instance, $\Qpub$ denotes the public-coin quantum \SMP~model and $\Qpri$ denotes the private-coin quantum \SMP~model. Unless otherwise specified, all protocols are private coin protocols. 

\textbf{One-Way Model.} In this model, Alice sends a single message to Bob, who should output $F(x,y)$ with probability at least $2/3$.  The cost of the protocol is the size of message Alice sends. By a classical result of Newman~\cite{newman1991private}, we can assume that all one-way protocols are private-coin protocols with an $O(\log n)$  additive overhead in the communication complexity. We use $\Rone$ to denote the one-way model of communication where Alice sends a classical message to Bob.

\textbf{Two-Way Model.} Here Alice and Bob are allowed to exchange messages, and Alice should finally output $F(x,y)$ with probability at least $ 2/3$. The cost of the protocol is the total size of the transcript. As before, by a result of Newman~\cite{newman1991private}, we can assume that all two-way protocols are private-coin protocols with an $O(\log n)$ additive overhead in the communication complexity. We use $\Rtwo$ to denote the model of two-way communication where Alice and Bob exchange classical~messages. 


We now discuss protocols where Alice and Bob can share an entangled state that is independent of their inputs. One important type of entangled state is the EPR pair: This is the state $\ketbra{ \Phi^+}{\Phi^+}$ where $\ket{\Phi^+}=\tfrac{1}{\sqrt{2}}\pbra{\ket{0}_A\ket{0}_B+\ket{1}_A\ket{1}_B}$. Here, the subscript on the qubit denotes which player owns the qubit. The upper bound in all our theorems will be established using quantum protocols where the players share a certain number of EPR pairs. We typically specify the dimension of the state shared by Alice and Bob.

\textbf{Protocols with entanglement.} We use $\Rent$ to denote the simultaneous model of communication where Alice and Bob share an entangled state and send classical messages to the referee. The model $\Qent$ is similarly defined, but Alice and Bob can send quantum messages to the~referee. We use $\Roneent$ to denote the one-way model where Alice and Bob share entanglement and Alice sends a classical message to Bob. We use $\Rtwoent$ to denote the two-way model where Alice and Bob share entanglement and the messages are classical.  In this model, by  teleportation , the players can exchange a limited number of qubits. Conversely, if the players can exchange a limited number of qubits, then Alice can create EPR pairs by herself and send the corresponding qubits to Bob.




For ease of readability, we summarize all the communication models in the tables below.

\begin{table}[h]
	\label{table:SMPmodels1}
	\begin{center}
		\begin{tabular}{ |c|c|c|c| } 
			\hline
			& Private Coins & Public Coins & Entanglement \\ 
			\hline
			Classical Messages & $\Rpri$ & $\Rpub$ & $\Rent$\\ 
			\hline
			Quantum Messages & $\Qpri$ & $\Qpub$ & $\Qent$ \\ 
			\hline
		\end{tabular}
	\end{center}
		\caption{Table of Simultaneous Communication Models}
\end{table}

\begin{table}[h!]
	\label{table:SMPmodels2}
	\begin{center}
		\begin{tabular}{ |c|c|c| } 
			\hline
			& One-way & Two-way \\ 
			& Private Coins & Private Coins\\
			& ($\equiv$ Public Coins) & ($\equiv$ Public Coins)\\
			\hline
			Classical Messages & $\Rone$ & $\Rtwo$  \\ 
			\hline
			Classical Messages & $\Roneent$ & $\Rtwoent$   \\ 
			\& Entanglement & &\\
			\hline
		\end{tabular}
	\end{center}
		\caption{Table of One-Way \& Two-Way Communication Models}
\end{table}
\subsection{XOR-Fibers of Communication Protocols}

\begin{definition}\label{definition:XOR_fiber}
	Given a communication protocol $C:\bin^n\times \bin^n\to [-1,1]$, the XOR-fiber of $C$ is a function $H:\bin^n\to [-1,1]$ defined at $z\in\bin^n$ by $H(z)=\E_{x\sim \bin^n}[C(x,x\odot z)]$, where $\odot$ denotes point-wise product.
\end{definition}

The communication complexity of XOR functions are well-studied and have connections to the log-rank conjecture, parity decision trees, lifting theorems and separations between quantum and classical communication complexity~\cite{MO09,HHL18,TWXZ13,SZ08,Zha13}.
	XOR-fibers of communication protocols naturally arise in the study of communication complexity of XOR functions. Although we are not aware of any published works defining the term ``XOR-fiber'', this concept has been studied in many works, most notably~\cite{girish2022quantum}  and~\cite{Raz95}.

\subsection{Fourier analysis}

\textbf{Fourier Analysis on the Boolean Hypercube.} We discuss some of the basics of Fourier analysis. Let $f:\bin^n\rightarrow \mathbb{R}$ be a function. The Fourier decomposition of $f$ is
$$
f(x)=\sum_{S\subseteq[n]} \widehat{f}(S)\chi_S(x),
$$
where $\chi_S(x)=\prod_{i\in S}x_i$. The \emph{Fourier coefficients} of $f$ are defined as $\widehat{f}(S)=\E_{x\sim\bin^n}[f(x) \cdot \chi_S(x)]$. For $\ell\in\Nbb,$ the level-$\ell$ Fourier mass of $f$ is denoted by $L_{1,\ell}(f)$ and is defined as follows.
\[ L_{1,\ell}(f)=\sum_{|S|=\ell}\abs{\widehat{f}(S)} \]
By Fourier growth bounds, we typically mean upper bounds on $L_{1,\ell}(f)$. We will need the following technical lemma, often called the level-$k$ inequality.
\begin{lemma}[{\cite[Lemma~10]{DBLP:conf/coco/Lee19}}]
	\label{lemma:level_k_inequality}
	Let $f:\{-1,1\}^n\rightarrow[-1,1]$ be a function with ${\E}_{x}[|f(x)|]=~\alpha$. Then for every $\ell\in\mathbb{N}$,
	\[
	\sum_{|S|=\ell}\widehat{f}(S)^2\leq 4\alpha^2\cdot (2e\cdot \ln(e/\alpha^{1/\ell}))^\ell.
	\]
\end{lemma}
Although \cite[Lemma~10]{DBLP:conf/coco/Lee19} is only stated for functions with range $[0,1]$, the same proof also applies for functions with range $[-1,1]$. We remark that the bound often invoked~\cite{o2014analysis,gavinsky2007exponential} is with the upper bound of $O(\alpha^2\cdot \ln^\ell(1/\alpha))$ (i.e., without the $1/\ell$ exponent on $\alpha$) which only holds for $\ell\le 2 \ln(1/\alpha)$. However this improved upper bound, proven recently in~\cite{DBLP:conf/coco/Lee19}, holds for all $\ell\in\Nbb$. This makes our proofs much simpler and saves some logarithmic~factors.

\textbf{Fourier analysis for Matrix-Valued Functions.} The Fourier coefficients of a matrix-valued function $f:\bin^n\to\mathbb{R}^{m\times m}$ are defined by
\begin{align*}
	\widehat{f}(S) := \mathop{\E}_{x\sim\01^n}\sbra{f(x)\cdot \chi_S(x)}
\end{align*}
for all $S\subseteq[n]$. We also use a matrix version of the level-$k$ inequality. 
\begin{lemma}[{\cite[Theorem~7]{ben2008hypercontractive}}] 
	\label{lemma:level_k_inequality_matrix_1}
	Let $\Hcal$ be a Hilbert space of dimension $2^c$ and let $f:\bin^n\to \Scal(\Hcal)$ be a density-matrix valued function. Then, for any $\ell\in \Nbb$ such that $\ell\le 2\ln(2)c,$
	\[
	\sum_{|S|=\ell} \Tr^2\big(\abs{\widehat{\rho}_S}\big)\le \pbra{{(2e\ln 2)\cdot c/\ell}}^\ell. 
	\]
\end{lemma}
Let $\Hcal'$ be a Hilbert space that contains $\Hcal$ of dimension $c'\ge c$ . We can view $f$ as a function $f:\bin^n\to \Scal(\Hcal')$ by simply appending zeroes to the output matrix. Given any $\ell\in\Nbb$ that is possibly larger than $c$, set $c':=c+\lceil\tfrac{\ell}{2\ln 2}\rceil$ so that $\ell\le 2\ln(2)\cdot c'$. Using this setting of parameters and \Cref{lemma:level_k_inequality_matrix_1}, we have the following corollary.
\begin{corollary}\label{lemma:level_k_inequality_matrix}
	Let $\Hcal$ be a Hilbert space of dimension $2^c$ and let $f:\bin^n\to \Scal(\Hcal)$ be a density-matrix valued function. Then, for any $\ell\in \Nbb$,
	\[ \sum_{|S|=\ell} \Tr^2\big(\abs{\widehat{\rho}_S}\big)\le O_\ell \pbra{(c/\ell)^{\ell}} + O_\ell(1). \]
	
\end{corollary}

%


\section{Proof of \texorpdfstring{\Cref{theoremRtwo}}{Theorem 1}}
\label{sec:firsttheoremproof}


In this section, we prove \Cref{lemma:main_lemma}, which we restate for convenience. 
\lemmaRtwo*

The proof of \Cref{theoremRtwo} follows from \Cref{lemma:main_lemma} and the techniques of~\cite{DBLP:conf/approx/GirishRZ21}. The quantum upper bound is presented in~\cite[Theorem 3.8]{DBLP:conf/approx/GirishRZ21}. For the lower bound, let $\Ccal$ be the set of $\Rtwoent$ protocols of cost at most $c$ using at most $d$ \equbits. Let $\Hcal$ be the set of XOR fibers of protocols in $\Ccal$. Applying ~\cite[Theorem 3.1]{DBLP:conf/approx/GirishRZ21} to $\Hcal$, it follows that the maximum advantage that protocols in $\Ccal$ have in solving the $\oplus^k$-Forrelation problem is at most ${O}\pbra{L_{1,2k}(\Hcal)\cdot n^{-k/2}}+{o}(n^{-k/2})$. By~\Cref{lemma:main_lemma}, this is at most $O\pbra{2^{5d}\cdot c^{2k}\cdot n^{-k/2}}$. Since $d\le k$, setting $c=\tau\cdot n^{1/4}$ for a small constant $\tau>0$ implies that this is at most $1/3$. The details of this calculation are deferred to~\Cref{sec:appendix_theorem_one}.

The rest of this section will be devoted to the proof of \Cref{lemma:main_lemma}. As described in the proof overview, this will consist of two parts. First, in \Cref{sec:decompose}, we show how to decompose entangled states as a linear combination ``simple'' states and next, in \Cref{sec:fourier_growth_rtwo}, we prove a Fourier growth bound on protocols that use ``simple'' states. We put these together in \Cref{sec:proof_main_lemma} to complete the proof of~\Cref{lemma:main_lemma}.

\label{sec:r2protocolsstructure}
\subsection{Decomposing an Arbitrary State as a Linear Combination of Simple States}
\label{sec:decompose}
\begin{definition}
	Let $\rho,\sigma$ be (possibly entangled) states in $\Scal(\Hcal_{A}\otimes H_{B})$. We say that $\rho$ is \equivalent to $\sigma$ if there exist unitaries  $U_A$ on $\Hcal_A$ and $V_B$ on $\Hcal_B$ such that $(U_A\otimes V_B)\rho(U_A\otimes  V_B)^\dagger=\sigma$. If $\sigma=\ketbra{ +}{ +}$ where $\ket{ +}=\frac{1}{\sqrt{2}}(\ket{0}_A \ket{0}_B+\ket{1}_A\ket{1}_B)$ corresponds to the EPR state, we say that $\rho$ is \equivalent to the EPR state and if $\sigma=\ketbra{0}{0}_A\otimes \ketbra{0}{0}_B$, we say that $\rho$ is \equivalent to the zero~state.
\end{definition}
We refer to the zero state and EPR state as simple. The main result of this section is as~follows.

\begin{claim}
	Let $d\in \Nbb$. Let $\Hcal_A,\Hcal_B$ be $2^d$-dimensional Hilbert spaces. Given a (possibly entangled) state $\rho$ in $\Scal(\Hcal_A\otimes \Hcal_B)$, we can express it as 
	\[
	\rho=\sum_{i=1}^{2^{4d}} \alpha_i \rho_i 
	\]
	where each $|\alpha_i|\le 2^d$ and $\rho_i\in \Scal(\Hcal_A\otimes \Hcal_B)$ is \equivalent to the zero state or the EPR state. \label{lemma:decomposition}
\end{claim}

This decomposition is useful as it simplifies the study of communication protocols with arbitrary entangled states. In particular, we have the following.

\begin{definition} We say that two communication protocols are \emph{equivalent} if on the same inputs, they produce the same distribution on transcripts. 
\end{definition}

\begin{fact}\label{claim:cc_equivalent}
	Let $\rho$ and $\sigma$ be two \equivalent states. Any communication protocol where Alice and Bob use $\rho$ as the entangled state can be transformed into an equivalent communication protocol where Alice and Bob use $\sigma$ as the entangled state. 
\end{fact}

\begin{proof} Let $U_A$ and $V_B$ be unitary operators on $\Hcal_A$ and $\Hcal_B$ respectively so that $(U_A\otimes V_B)\rho(U_A \otimes V_B)^\dagger=\sigma.$ We modify the protocol by first having Alice and Bob apply $U_A$ and $V_B$ respectively to their parts of the entangled state, then apply $U_A^{-1}$ and $V_B^{-1}$ and then continue as per the original protocol. Observe that this transformation does not change the distribution on the~transcripts. The first step of the new protocol transforms $\rho$ into~$\sigma$. Thus, we may think of the new protocol as an equivalent communication protocol where the initial entangled state is $\sigma$. 
\end{proof}

We now turn to proving~\Cref{claim:structure_rtwoent}. We will use the following fact.

\begin{fact} Let $\ket{\phi}\in \Hcal_A$ and $\ket{\psi}\in \Hcal_B$ be unit vectors. Then, the state $\rho:=\ketbra{\phi}{\phi}_A\otimes\ketbra{\psi}{\psi}_B$ is \equivalent to the zero state. \label{fact:tensor}
\end{fact}
\begin{proof}
	Consider unitaries $U_A$ and $V_B$ such that $U_A\ket{\phi}=\ket{0}$ and $V_B\ket{\psi}=\ket{0}$. Observe that $(U_A\otimes V_B) \rho (U_A\otimes V_B)^\dagger=\ketbra{0}{0}_A\otimes\ketbra{0}{0}_B$ and hence $\rho$ is \equivalent to the zero state.
\end{proof}


We now complete the proof of~\Cref{lemma:decomposition}.
\begin{proof}[Proof of \Cref{lemma:decomposition}]
	Let $i,j\in [2^{2d}]$ be such that $i\neq j$. Define $\rho^{(i,j)}\in\Scal(\Hcal_A\otimes \Hcal_B)$ and $\alpha_{i,j}\in \Rbb$~by 
	\[
	\rho^{(i,j)}:=\tfrac{1}{2}\pbra{\ket{i}+\ket{j}}\pbra{\bra{i}+\bra{j}} \quad\text{and}\quad \alpha_{i,j}=\rho_{i,j},
	\]
	where $\rho_{i,j}$ is the $(i,j)$th entry of $\rho$. 
	Define $\rho^{(i,i)}\in \Scal(\Hcal_A\otimes \Hcal_B)$ and $\alpha_{i,i}\in \Rbb$ by
	\[ \rho^{(i,i)}:=\ket{i}\bra{i}\quad\text{and}\quad\alpha_{i,i}=\rho_{i,i}-\sum_{j\in[2^{2d}],j\neq i} \rho_{i,j}.
	\] 
	Observe that $\sum_{i, j \in[2^{2d}]} \alpha_{i,j}\rho^{(i,j)}=\rho$. Furthermore, for $i\neq j\in[2^{2d}]$, we have $\abs{\alpha_{i,j}}=\abs{\rho_{i,j}}\le 1$ and $\abs{\alpha_{i,i}}\le \sum_{j\in [2^{2d}]}\abs{\rho_{i,j}}\le \sqrt{2^{2d}}\cdot \sqrt{\sum_{j\in[2^{2d}]}\rho_{i,j}^2}\le 2^{d}$ due to Cauchy-Schwarz and the fact that $\rho$ is a quantum state. It suffices to argue that $\rho^{(i,j)}$ is \equivalent to the zero state or the EPR~state.  Firstly, every $i\in[2^{2d}]$ can be uniquely identified with $(a,b)$ where $a,b\in[2^{d}]$. Similarly, each $j\in[2^{2d}]$ can be uniquely identified with $(p,q)$ where $p,q\in[2^{d}]$. Consider the following cases. 
	
	\textbf{Case 1:} Suppose $a=p$ and $b=q$ (or equivalently $i=j$), then $\rho^{(i,j)}= \ketbra{a}{a}_A \otimes \ketbra{b}{b}_B$. By \Cref{fact:tensor}, $\rho^{(i,j)}$ is \equivalent to the zero state.
	
	\textbf{Case 2:} Suppose $a=p$ and $b\neq q$, then \begin{align*}
		\rho^{(i,j)}&\triangleq \tfrac{1}{2}\pbra{\ket{a}_A \ket{b}_B+\ket{p}_A \ket{q}_B}\pbra{\bra{a}_A \bra{b}_B+\bra{p}_A \bra{q}_B}\\ &=\ketbra{a}{a}_A \otimes  \pbra{\tfrac{ \ket{b}_B+ \ket{q}_B}{\sqrt 2}}  \pbra{\tfrac{\bra{b}_B+ \bra{q}_B}{\sqrt 2}}.
	\end{align*} 
	Since $\ket{b}$ and $\ket{q}$ are orthogonal, $\tfrac{\ket{b}+\ket{q}}{\sqrt{2}}$ is a unit vector and by \Cref{fact:tensor}, $\rho^{(i,j)}$ is \equivalent to the zero state.
	
	\textbf{Case 3:} Suppose $a\neq p$ and $b\neq q$. Let $U_A$ be any unitary operator that maps $\ket{a}$ to $\ket{0}$ and $\ket{p}$ to $\ket{1}$. This is well defined since $a\neq p$. Similarly define $V_B$ to be any unitary operator that maps $\ket{b}$ to $\ket{0}$ and $\ket{q}$ to $\ket{1}$. This is well defined since $b\neq q$. Observe that 
	\begin{align*} (U_A\otimes V_B)\rho^{(i,j)}(U_A\otimes V_B)^\dagger &= \tfrac{1}{2}(\ket{0}_A\ket{0}_B+\ket{1}_A\ket{1}_B)(\bra{0}_A\bra{0}_B+\bra{1}_A\bra{1}_B). 
	\end{align*}
	Thus, $\rho^{(i,j)}$ is \equivalent to the EPR state. This completes the proof of~\Cref{lemma:decomposition}.
\end{proof}

\subsection{Fourier Growth Bounds on \texorpdfstring{$\Rtwoent$}{R||*} protocols with the EPR State}
\label{sec:fourier_growth_rtwo}
We now prove a Fourier growth bound on $\Rtwoent$ protocols where Alice and Bob share a single EPR pair. Note that such protocols can easily simulate protocols that share the zero state, since Alice and Bob can simply produce $\ketbra{0}{0}_A$ and $\ketbra{0}{0}_B$ without communication and ignore the EPR pair. The main technical contribution of this subsection is the following lemma.

\begin{lemma}\label{lemma:second_main_lemma}
	Let $C:\bin^n\times \bin^n\to [-1,1]$ be a $\Rtwoent$ protocol of cost $c$ where the players share the EPR state. Let
	$H$ be its XOR-fiber as in~\Cref{definition:XOR_fiber}. Then, for all $\ell\in \Nbb$,~we~have
	\[L_{1,\ell}(H)\le O_\ell(1)+O_\ell\pbra{(c/\ell)^\ell}.
	\]
\end{lemma}
The following claim helps understand the acceptance probability of $\Rtwoent$~protocols.

\begin{claim}\label{claim:structure_rtwoent} Let $C:\bin^n\times\bin^n\to [-1,1]$ be any $\Rtwoent$ protocol of cost $c$ where Alice and Bob share a state $\rho\in \Scal(\Hcal_A\otimes \Hcal_B)$ where $\Hcal_A$ and $\Hcal_B$ are Hilbert spaces of dimension $2^d$. Then, there exists $m\in \Nbb$ and Hilbert spaces $\Hcal'_A,\Hcal'_B$ of dimension $m$ such that the expected output of the protocol on inputs $x,y\in\bin^n$ can be expressed as 
	\[C(x,y)=\sum_{z\in \bin^c}\Tr\pbra{(E_z(x)\otimes F_z(y))\cdot \rho'}\cdot (-1)^{\indi[z\in A]}, \]
	where for all $x,y\in\bin^n$ and $z\in\bin^c$,
	\begin{enumerate}
		\item $E_z(x)$ is a $2^{d+m}\times 2^{d+m}$ positive semidefinite matrix acting on $\Hcal_A\otimes \Hcal'_A$,
		\item $F_z(y)$ is  a $2^{d+m}\times 2^{d+m}$ positive semidefinite matrices acting on $\Hcal_B\otimes \Hcal'_B$, 
		\item  $\sum_{z'\in\bin^c} E_{z'}(x)\otimes F_{z'}(y)=\mathbb{I}$, 
		\item $\rho'=\rho\otimes \ket{0^m}\bra{0^m}_A\otimes \ket{0^m}\bra{0^m}_B$, and
		\item  $A\subseteq\bin^c$ is some subset.
	\end{enumerate} 
\end{claim}
We defer the proof of this claim to Appendix~\ref{app:claim2.2proof}. We now prove \Cref{lemma:second_main_lemma} using \Cref{claim:structure_rtwoent}.

\begin{proof}[Proof of \Cref{lemma:second_main_lemma} using \Cref{claim:structure_rtwoent}]
	
	We use $\ketbra{ \Phi^+}{ \Phi^+}$ to denote the single EPR state. We use \Cref{claim:structure_rtwoent} which describes the structure of arbitrary $\Rtwoent$ protocols of cost $c$. By \Cref{claim:structure_rtwoent}, the expected output of the protocol $C$ on inputs $x,y\in\bin^n$ can be expressed as 
	\[
	C(x,y)=\sum_{z\in \bin^c} \Tr((E_z(x)\otimes F_z(y)\cdot \rho)\cdot (-1)^{\indi[z\in A]}
	\] where 
	$A\subseteq\bin^c$, and $E_z(x)$ and $F_z(y)$ are positive semidefinite operators such that \begin{align}\label{eq:sum_to_identity} \sum_{z'\in \bin^c}E_{z'}(x)\otimes F_{z'}(y)=\mathbb{I}, \text{ and }
		\rho=\ketbra{ \Phi^+}{ \Phi^+}\otimes \ketbra{0^m}{0^m}_A\otimes \ketbra{0^m}{0^m}_B\end{align} 
	for some parameter $m\in \Nbb$.
	Recall that by the definition of XOR-fiber, for all $w\in\bin^n$, we have $H(w)=\E_{x\sim \bin^n} [C(x,x\odot w)]$. Hence,
	\begin{align*}
		&\widehat{H}(S)\\&=\underset{w\sim \bin^n}{\E} [H(w)\chi_S(w)]\\
		&=\underset{w,x\sim \bin^n}{\E}\sbra{\sum_{z\in \bin^c}  \Tr\pbra{\pbra{E_z(x)\otimes F_z(x\odot w)}\cdot \rho}\cdot  \chi_S(w) \cdot(-1)^{\indi[z\in A]}}\\
		&=\underset{w,x}{\E}\sbra{\sum_{z\in \bin^c}\sum_{T,Q}  \Tr\pbra{\pbra{\widehat{E_z}(T)\otimes \widehat{F_z}(Q)}\cdot \rho}\cdot  \chi_{S+Q}(w)\cdot \chi_{T+Q}(x)\cdot(-1)^{\indi[z\in A]}}\\
		&=\sum_{z\in \bin^c} \Tr\pbra{\pbra{\widehat{E_z}(S)\otimes \widehat{F_z}(S)}\cdot \rho} \cdot(-1)^{\indi[z\in A]}.
	\end{align*}
	Our goal is to upper bound 
	\begin{align}
		\label{eq:goal}\begin{split}
			L_{1,\ell}(H)\triangleq \sum_{|S|=\ell}\abs{ \widehat{H}(S)}&=\sum_{|S|=\ell}\abs{ \sum_{z\in \bin^c} \Tr\pbra{\pbra{\widehat{E_z}(S)\otimes \widehat{F_z}(S)}\cdot \rho} \cdot(-1)^{\indi[z\in A]}}\\
			&\le \sum_{z\in \bin^c}\sum_{|S|=\ell} \abs{\Tr\pbra{\pbra{\widehat{E_z}(S)\otimes \widehat{F_z}(S)}\cdot \rho}}.
		\end{split}
	\end{align} 
	By Eq.~\eqref{eq:sum_to_identity}, the density matrix of the state $\rho$ has exactly four non-zero coordinates.
	This zeroes all but four coordinates of $\widehat{E_z}(S)\otimes \widehat{F_z}(S)$ in the R.H.S. of Eq.~\eqref{eq:goal}. More precisely, we have
	\[ \Tr\pbra{\pbra{\widehat{E_z}(S)\otimes \widehat{F_z}(S)}\cdot \rho}= \frac{1}{2}\cdot\sum_{i,j\in\{1,2^m+1\}} \widehat{E_z}(S)[i,j]\cdot \widehat{F_z}(S)[i,j]. \]
	Substituting this in Eq.~\eqref{eq:goal},we have
	\begin{align*} 
		L_{1,\ell}(H)&\le \sum_{z\in \01^c} \sum_{|S|=\ell}\abs{ \sum_{i,j\in\{1,2^m+1\}} \widehat{E_z}(S)[i,j]\cdot \widehat{F_z}(S)[i,j]} \\
		&\le\sum_{z\in \bin^c} \sum_{i,j\in\{1,2^m+1\}} \sum_{|S|=\ell}  \abs{\widehat{E_z}(S)[i,j]\cdot \widehat{F_z}(S)[i,j]} \\
		&\le \sum_{z\in \bin^c} \sum_{i,j\in\{1,2^m+1\}}\sqrt{ \sum_{|S|=\ell}  \widehat{E_z}(S)[i,j]^2  } \cdot \sqrt{ \sum_{|S|=\ell} \widehat{F_z}(S)[i,j]^2}.
	\end{align*}
	
	For $i\neq j\in\{1,2^m+1\}$, let $e_z[i,j]:=\E_x\big[\abs{E_z(x)[i,j]}\big]$ and $f_z[i,j]:=\E_y \big[\abs{F_z(y)[i,j]}\big]$ (where the expectation is over the uniform distribution on $\bin^n$). Similarly, let $e_z[i,i]=\E_x\big[E_z(x)[i,i]\big]$ and $f_z[i,i]=\E_y\big[F_z(y)[i,i]\big]$. Since $E_x$ and $F_y$ are positive semidefinite, their diagonal entries are non-negative. Using the level-$k$ inequality (\Cref{lemma:level_k_inequality}) we get that for all $i,j\in\{1,2^m+1\}$,
	\begin{align*}
		\sum_{|S|=\ell}  \widehat{E_z}(S)[i,j]^2  &\le   4 (\E_x\sbra{ |E_z(x)[i,j]|})^2\cdot \pbra{2e \ln\pbra{e/(\E_x\sbra{ |E_z(x)[i,j]|})^{1/\ell}}}^\ell\\
		&=4(e_z[i,j])^2\cdot \pbra{2e \ln\pbra{e/(e_z[i,j])^{1/\ell}}}^\ell.
	\end{align*}
	We can now upper bound $L_{1,\ell}(H)$ as follows:
	\begin{align*} 
		&L_{1,\ell}(H)\\&\leq 4 \sum_{\substack{z\in \bin^c\\i,j\in\{1,2^m+1\}}}  e_z[i,j] \cdot f_z[i,j] \cdot \pbra{2e \ln \pbra{\frac{e}{e_z[i,j]^{1/\ell}}}}^{\ell/2}  \cdot \pbra{2e \ln \pbra{\frac{e}{f_z[i,j]^{1/\ell}}}}^{\ell/2}\\
		&\le 4 \sum_{\substack{z\in \bin^c\\i,j\in\{1,2^m+1\}}} e_z[i,j] \cdot f_z[i,j] \cdot\pbra{2e \ln \pbra{\frac{e^2}{e_z[i,j]^{1/\ell}\cdot f_z[i,j]^{1/\ell}}}}^{\ell}. 
	\end{align*}
	We now use the concavity of the function $h(\gamma)=\gamma\cdot \ln(e^2/\gamma^{1/\ell})^{\ell}$ for $\gamma\in [0,1]$.\footnote{The concavity of this function is proved in~\cite[Claim~16]{DBLP:conf/coco/Lee19}. In more detail, observe that $h(\gamma)=\gamma\cdot \ln\pbra{(e/\gamma^{1/(2\ell)})^2}^{\ell}=\gamma\cdot \ln\pbra{e/\gamma^{1/(2\ell)}}^{\ell}\cdot 2^{\ell}$ which by their notation, equals $2^{\ell}\cdot \phi_{2\ell}(\gamma)$. It is shown that for all positive $\ell,$ the function $\phi_{2\ell}(\gamma)$ is concave and increasing for $\gamma\in [0,1]$.} Jensen's inequality implies that for any $p_z\in[0,1]$, we have that $\E_z[h(p_z)]\leq h(\E_z[p_z])$.
	Setting $p_z=e_z[i,j]\cdot f_z[i,j]$, we conclude that $L_{1,\ell}(H)$ is at most
	\begin{align*}
	 &2^{c+2} \sum_{i,j\in \{1,2^m+1\}} \pbra{\underset{z\in \bin^c}{\E}\sbra{ e_z[i,j] f_z[i,j]}}  \pbra{2e \ln \pbra{\frac{e^2}{\pbra{\E_{z\in \bin^c} \big[e_z[i,j] f_z[i,j]\big]}^{1/\ell}}}}^{\ell}\\
		&=4 \sum_{i,j\in \{1,2^m+1\}} \pbra{\sum_{z\in \bin^c} e_z[i,j] f_z[i,j]}  \pbra{2e \ln \pbra{\frac{e^2\cdot 2^{c/\ell}}{\pbra{\sum_{z\in \bin^c} e_z[i,j] f_z[i,j]}^{1/\ell}}}}^{\ell}
		.\end{align*}
	To simplify notation, for $i,j\in\{1,2^m+1\},$ we define $\beta_{i,j}\in \Rbb$ by $\beta_{i,j}:=\sum_{z\in\bin^c} e_z[i,j]\cdot f_z[i,j]$. With this notation, we have
	\begin{align*}
		L_{1,\ell}(H)&\le 4\sum_{i,j\in \{1,2^m+1\}} \beta_{i,j} \cdot  \pbra{2e  \ln \pbra{{\frac{e^2\cdot 2^{c/\ell}}{\beta_{i,j}^{1/\ell}}}}}^{\ell}\\
		&\le 4\sum_{i,j\in \{1,2^m+1\}} \beta_{i,j} \cdot  \pbra{2e  \ln \pbra{{\frac{e^2}{\beta_{i,j}^{1/\ell}} }   }+\frac{2e c}{\ell}}^{\ell}\\
		&\le 4\sum_{i,j\in \{1,2^m+1\}} \beta_{i,j} \cdot 2^{\ell}\cdot  \pbra{2e\ln \pbra{{\frac{e^2}{\beta_{i,j}^{1/\ell}} }   }}^{\ell} +4\sum_{i,j\in \{1,2^m+1\}} \beta_{i,j} \cdot 2^{\ell}\cdot  \pbra{\frac{2e c}{\ell}}^{\ell}.
	\end{align*}

	For all $x,y\in\bin^n,$ the matrix $E_z(x)\otimes F_z(y)$ is positive semidefinite. Furthermore, for any positive semidefinite matrix $M$, we have $\abs{M_{i,j}}\le \tfrac{1}{2}\pbra{M_{i,i}+M_{j,j}}$ for all $i\neq j$. This implies that for all $x,y\in\bin^n$ and $i\neq j\in\{1,2^m+1\}$,
	\[ \abs{E_z(x)[i,j]}\cdot  \abs{F_z(y)[i,j]} \le \tfrac{1}{2} \pbra{  E_z(x)[i,i]\cdot  F_z(y)[i,i]+  E_z(x)[j,j] \cdot F_z(y)[j,j]} \]
	Taking an expectation over $x,y\sim\bin^n$ implies that for $i\neq j\in\{1,2^m+1\}$,
	\[
	e_z[i,j]\cdot f_z[i,j]\le \tfrac{1}{2}\pbra{ e_z[i,i]\cdot f_z[i,i]+e_z[j,j]\cdot f_z[j,j]} .
	\]
	By Eq.~\eqref{eq:sum_to_identity}, since $\sum_{z\in\bin^c} E_z(x)\otimes F_z(y)=\mathbb{I}$ for $x,y\in \bin^n$, we have for all $i\neq j\in\{1,2^m+1\}$
	\[ \beta_{i,j}\triangleq \sum_{z\in\bin^c} e_z[i,j]\cdot f_z[i,j] \le \sum_{z\in\bin^c}\tfrac{1}{2}\pbra{ e_z[i,i]\cdot f_z[i,i]+e_z[j,j]\cdot f_z[j,j]}= 1,
	\] 
	where in the final equality we used that $e_z[i,i]=\E_x\big[E_z(x)[i,i]\big]$. Thus, we have 
	\[ 
	L_{1,\ell}(H)\le O_\ell(1)\cdot \sum_{i,j\in \{1,2^m+1\}} \beta_{i,j}\cdot \pbra{2e  \ln \pbra{{\frac{e^2}{\beta_{i,j}^{1/\ell}} }   }}^{\ell} + O_\ell\pbra{(c/\ell)^{\ell}}
	.\]
	It follows from~\cite[Claim~16]{DBLP:conf/coco/Lee19} that the function $h(\gamma)=\gamma\ln (e^2/\gamma^{1/\ell})^{\ell}$ is increasing for $\gamma\in[0,1]$ and the value at $\gamma=1$ is $2^\ell$.
	Thus, we have
	\begin{align*}
		\beta_{i,j}\cdot  \ln \pbra{\frac{e^2}{\beta_{i,j}^{1/\ell}}}^{\ell}&\le  2^{\ell}.
	\end{align*}
	Therefore,
	\[L_{1,\ell}(H)\leq O_\ell(1)+ O_\ell\big((c/\ell)^{\ell}\big),
	\]
	proving \Cref{lemma:second_main_lemma}.

	\subsection{Putting Things Together}
	
	\label{sec:proof_main_lemma}
	We now prove \Cref{lemma:main_lemma} using  \Cref{lemma:decomposition} and \Cref{lemma:second_main_lemma}. Consider any interactive randomized communication protocol $C_\rho$ of cost $c$ that uses $\rho$ as the entangled state, where $\rho$ is an arbitrary state on $2d$ qubits for a parameter $d\in \Nbb$. Consider the decomposition 
	\[
	\rho=\sum_{i=1}^{2^{4d}} \alpha_{i} \rho_i
	\] as given by \Cref{lemma:decomposition}. Let $C_\rho(x,y)$ denote the expected output of the protocol on inputs $x,y$ as before. Observe that $C_\rho(x,y)$ is linear in $\rho$ due to~\Cref{claim:structure_rtwoent}. Thus, the expected output of the protocol can be expressed as \[C_\rho(x,y)=\sum_{i=1}^{2^{4d}} \alpha_{i}C_{\rho_i}(x,y).\] 
	By \Cref{claim:cc_equivalent}, we have $C_{\rho_i}(x,y)=C^{(i)}_{\sigma_i}(x,y)$ where $\sigma_i$ is either the zero state or the EPR state and $C^{(i)}$ is some $\Rtwoent$ protocol that is equivalent to $C$ and uses $\sigma_i$ as the shared state. Therefore, we have for all $z\in \bin^n,$
	\[H(z)=\sum_{i=1}^{2^{4d}}\alpha_{i}H^{(i)}_{\sigma_i}(z)\] 
	where $H^{(i)}_{\sigma_i}$ is the XOR-fiber of $C^{(i)}_{\sigma_i}$. In particular,
	\begin{align*}
		L_{1,\ell}(H)&\le \sum_{i=1}^{ 2^{4d}}\abs{\alpha_{i}}\sum_{|S|=\ell}\abs{\widehat{H^{(i)}_{\sigma_i }}(S)}\le 2^{5d}\cdot \max_{i\in[2^{4d}]}\pbra{ L_{1,\ell}\pbra{ H^{(i)}_{\sigma_i} }}. \end{align*} 
	
	Here, we use the fact that $\abs{\alpha_{i}}\le 2^{d}$. Since each $\sigma_i$ is either the zero state or the EPR state and $C^{(i)}$ has cost at most $c$, we can apply \Cref{lemma:second_main_lemma} to conclude that
	\[
	L_{1,\ell}\pbra{H^{(i)}_{\sigma_i }} \leq  O_\ell(1)+ O_\ell\big((c/\ell)^\ell\big) \le O(c^\ell).
	\]
	This proves \Cref{lemma:main_lemma}.
\end{proof}

\section{Proof of \texorpdfstring{\Cref{theoremQparallel}}{Theorem 3}}
\label{sec:secondthmproof}

The main technical contribution in this section is a Fourier growth bound on $\Qpub$ protocols. We restate~\Cref{lemma:main_lemma_2} for convenience.

\lemmaQpri*
\begin{proof}[Proof of \Cref{theoremQparallel} using \Cref{lemma:main_lemma_2}]
	Let $c=\tau\cdot n^{1/4}$ for a small enough constant $\tau>0$. Let $\Hcal$ be the family of XOR-fibers of $\Qpub$ protocols of cost at most $c$. This is a restriction-closed family of Boolean functions. The results of~\cite{girish2022quantum} imply that the maximum advantage that functions in $\Hcal$ have in computing the Forrelation problem is at most $O\pbra{\tfrac{L_{1,2}(H)}{\sqrt{n}}}+O\pbra{\tfrac{1}{\sqrt{n}}}$. Using~\Cref{lemma:main_lemma_2}, this quantity is at most $\tfrac{c^2}{\sqrt{n}}\le \tau^2\ll 1$. This, along with the techniques of~\cite{girish2022quantum} implies that $\Qpub$ protocols solving the Forrelation problem require communication cost ${\Omega}(n^{1/4})$. 
\end{proof}

We now prove \Cref{lemma:main_lemma_2}. 
\begin{proof}[Proof of \Cref{lemma:main_lemma_2}] Firstly, it suffices to prove the lemma for $\Qpri$ protocols, by the fact that $\Qpub$ protocols are defined by expectations over $\Qpri$ protocols and by triangle inequality. Let $C$ be any $\Qpri$ protocol of cost $c$ where Alice and Bob don't share entanglement. Without loss of generality, the protocol has the following form. Let $\Hcal$ be a Hilbert space of dimension $2^c$. Alice sends Charlie $\rho(x)$ and Bob sends Charlie $\sigma(y)$ where $\rho(x),\sigma(y)\in \Scal(\Hcal)$. Then Charlie applies a two-outcome POVM $\{M_{1},M_{-1}\}$  and announces the outcome as the answer. It is not too hard to see that the expected output of the protocol is precisely $C(x,y)=\Tr(E\cdot(\rho(x)\otimes \sigma(y))$ where $E=M_1^\dagger M_1-M_{-1}^\dagger M_{-1}$.
	Observe that for all $S\subseteq[n]$, 
	\[
	\widehat{H}(S)\triangleq  \mathop{\E}_{x,y\sim\bin^n}\sbra{\Tr(E\cdot(\rho(x)\otimes \sigma(y)))\cdot \chi_S(x\odot y)}=\Tr(E\cdot( \widehat{\rho}_S\otimes \widehat{\sigma}_S)).
	\]
	Therefore, we have
	\begin{align}\label{eq:quantum_goal}
		\begin{split}
			\sum_{|S|=\ell}\abs{\widehat{H}(S)}&=\sum_{|S|=\ell}\abs{\Tr\pbra{E\cdot \pbra{\widehat{\rho}_S\otimes \widehat{\sigma}_S}}}\\
			&\le 2 \sum_{|S|=\ell}\Tr\pbra{\abs{\widehat{\rho}_S}} \cdot \Tr\pbra{\abs{\widehat{\sigma}_S}} 
			\le  \sqrt{\sum_{|S|=\ell}\Tr\pbra{\abs{\widehat{\rho}_S}}^2 } \cdot\sqrt{\sum_{|S|=\ell} \Tr\pbra{\abs{\widehat{\sigma}_S}}^2} .
		\end{split} 
	\end{align}
	Here, we used \Cref{fact:inequality} on $M_1^\dagger M_1$ and $M_{-1}^\dagger M_{-1}$. We now apply \Cref{lemma:level_k_inequality_matrix} to the density-matrix valued functions $\rho(x),\sigma(y)$ to conclude that
	\[
	\max\Big\{\sum_{|S|=\ell} \Tr^2\pbra{\abs{\widehat{\rho}_S}},\sum_{|S|=\ell} \Tr^2\pbra{\abs{\widehat{\sigma}_S}}\Big\}\le O_\ell \pbra{(c/\ell)^\ell} + O_\ell(1) .
	\]
	Substituting this in Eq.~\eqref{eq:quantum_goal}, we have 
	\begin{align*}
		L_{1,\ell}(H)&\le  O_\ell \pbra{(c/\ell)^\ell}+O_\ell(1) \le O_\ell(c^\ell).
	\end{align*}
	This proves the desired upper bound. \end{proof}

\begin{remark}
	\Cref{lemma:main_lemma_2}, along with~\cite{DBLP:conf/approx/GirishRZ21} implies that the advantage that $\Qpub$ protocols of cost ${o}(n^{1/4})$ have in computing the $\oplus^k$-Forrelation problem is at most $\exp(-\Omega(k))$. Using this, along with~\Cref{lemma:simulation_q}, one can show that $\Qent$ protocols of cost $o(n^{1/4})$ where Alice and Bob only share $O(k)$ \equbits~cannot solve the $\oplus^k$-Forrelation problem. On the other hand, \cite{DBLP:conf/approx/GirishRZ21} shows that if Alice and Bob share $\tilde{O}(k^5\log^3n)$ EPR pairs, then the $\oplus^k$-Forrelation problem can be solved by $\Qent$ protocols of cost $\tilde{O}(k^5\log^3n)$. This gives a separation between $\Qent$ with more entanglement and $\Qent$ with less entanglement.
\end{remark}

\section{Proof of \texorpdfstring{\Cref{theoremRoneQsmp}}{Theorem 3}}
\label{sec:thirdthmproof}
For convenience we restate \Cref{theoremRoneQsmp} below.

\theoremRoneQsmp*

\begin{proof}
	The proof of this theorem proceeds in two steps. First, we assume by contradiction that there exists an $\Roneent$ or $\Qent$ protocol of cost $c$ that uses $2d$ \equbits~and solves the $\oplus^k$-Boolean Hidden Matching problem with advantage at least $\tfrac{1}{3}$. Using this, we produce an $\Rone$ or $\Qpri$ protocol respectively that solves the $\oplus^k$-Boolean Hidden Matching problem with cost at most $c+O(d)$; this protocol uses no entanglement but only succeeds with advantage at least $2^{-\Theta(d)},$ that is, the success probability is at least $ \tfrac{1}{2}+2^{-\Theta(d)}.$ This will be proved in \Cref{lemma:simulation_q} and \Cref{lemma:simulation_r}. We then argue that $\Rone$ and $\Qpri$ protocols satisfy an XOR lemma with respect to the Boolean Hidden Matching problem. That is, the advantage that cost $c$ protocols have in solving the $\oplus^k$-Boolean Hidden Matching problem is at most $O_k\pbra{\frac{(c/k)^3}{n}}^{k/2}$ for the $\Qpri$ model and at most $O_k \pbra{\frac{(c/k)^{2}}{n}}^k$ for the $\Rone$ model. This will be proved in \Cref{lemma:main_lemma_3} and \Cref{lemma:main_lemma_4} respectively. Combining the aforementioned lemma, for the $\Qent$ case and $\Roneent$ case, we have \begin{equation}\label{eq:contradiction}\frac{1}{3}\cdot 2^{-4d}\le O_k\pbra{\frac{((c+2d)/k)^3}{n}}^{k/2}\quad\text{and}\quad\frac{1}{6}\cdot 2^{-4d}\le O_k\pbra{\frac{((c+10d)/k)^2}{n}}^{k/2}\end{equation}
	respectively. Let $\tau,\tau'>0$ be sufficiently small constants and $d=\tau'\cdot k$. For the $\Qent$ model, we can set $c=\tau\cdot kn^{1/3}$ and for the $\Roneent$ model set $c=\tau\cdot kn^{1/2}$ so  that Eq.~\eqref{eq:contradiction} is violated. This proves our communication lower bound. The quantum upper bound is presented in~\Cref{sec:appendix_theorem_three}.
\end{proof}

\subsection{Removing Entanglement from Protocols}

Our main contributions in this section are the following lemmas. Below, we assume $d$ is large~enough. 

\begin{lemma} \label{lemma:simulation_q}
	Given any communication protocol for computing $F$ with cost $c$ and advantage $\tfrac{1}{3}$ in the $\Qent$ model where Alice and Bob share at most $2d$ \equbits, there is a protocol of cost $c+ 2d$ in the $\Qpri$ model that computes $F$ with advantage $ \tfrac{1}{3}\cdot 2^{-4d}.$
\end{lemma}

\begin{lemma} \label{lemma:simulation_r} Given any communication protocol for computing $F$ with cost $c$ and advantage $\tfrac{1}{3}$ in the $\Roneent$ model where Alice and Bob share at most $2d$ \equbits, there is a protocol of cost $c+ 10d$ in the $\Rone$ model that computes $F$ with advantage $\tfrac{1}{6}\cdot 2^{-4d}.$
\end{lemma}

\begin{proof}[Proof of \Cref{lemma:simulation_q}]
	Let $C$ be a communication protocol of cost $c$ in the $\Qent$ model that uses $2d$ \equbits. Without loss of generality, the protocol has the following form. Let $\Hcal'_A,\Hcal'_B$ be Hilbert spaces of dimension $2^c$ each. Alice applies a  quantum channel $U_x:\Scal(\Hcal_A)\to \Scal(\Hcal'_A)$ to her part of the shared state and Bob applies a  quantum channel $V_y:\Scal(\Hcal_B)\to \Scal(\Hcal'_B)$ to his part of the shared state. The players then send the resulting states to the referee. The referee evaluates a two-outcome POVM on the received state and returns the outcome as the answer. In particular,~if 
	$$
	\rho:=\sum_{a,b,p,q\in [2^d]}\rho_{a,b,p,q}\ketbra{a}{p}_A\otimes \ketbra{b}{q}_B
	$$
	is the state shared by Alice and Bob, then the state received by the referee is
	\begin{align}
		\label{eq:statereceivedbyQ||star}
		\sum_{a,b,p,q\in [2^d]}\rho_{a,b,p,q}U_x (\ket{a}\bra{p})\otimes V_y (\ket{b}\bra{q}).
	\end{align}
	We now produce a $\Qpri$ protocol where Charlie receives the state in Eq.~\eqref{eq:statereceivedbyQ||star} with probability  $2^{-4d}$, furthermore, he knows when this state is received. When Charlie receives this state, he continues as per the original protocol and when he does not receive this state, he returns a uniformly random bit. This would produce a $\Qpri$ protocol that computes $F$ with advantage $\tfrac{1}{3}\cdot 2^{-4d}$. The protocol is as follows. Consider a $\Qpri$ protocol where Alice and Bob create the states 
	\[\sum_{a,b,p,q\in[2^d]}\rho_{a,b,p,q}\ketbra{a}{p}_A\otimes \ketbra{b}{q}_A \quad\text{and}\quad 2^{-2d}\sum_{b',q'\in[2^d]}  \ketbra{b'}{q'}_B\otimes  \ketbra{b'}{q'}_B\quad \text{respectively}.\] 
	Alice applies $U(x)$ to the \emph{first half} of the qubits of her state and sends the entire state to Charlie. Bob applies $V(y)$ to the \emph{second half} of the qubits of his state and sends the entire state to Charlie. The state received by Charlie is
	\[ 
	2^{-2d}\sum_{\substack{a,b,b'\in [2^d]\\p,q,q'\in[2^d]}}\rho_{a,b,p,q} U_x\pbra{\ket{a}\bra{p}}\otimes  \ket{b}\bra{q}\otimes { \ket{b'}\bra{q'}\otimes V_y\pbra{\ket{b'}\bra{q'}}}.
	\]
	Charlie projects onto states such that $b=b'$ and $q=q'$. More precisely, Charlie considers the measurement operator on the qubits from $d+1$ to $3d$ defined by projection onto $\{\ket{b}\bra{q}\otimes \ket{b}\bra{q}:b,q\in[2^d]\}$. This measurement operator applied to the above state produces the state
	\[
	\sum_{a,b,p,q,\in[2^d]}\rho_{a,b,p,q} U_x\pbra{\ket{a}\bra{p}}\otimes  \ket{b}\bra{q}\otimes \ket{b}\bra{q}\otimes V_y\pbra{\ket{b}\bra{q}}.
	\]
	with probability $1/2^{2d}$, furthermore, Charlie can tell when this state is produced using the measurement outcome. Charlie then applies the Hadamard operator to the qubits $d+1,\ldots,3d$ and measures those qubits in the standard basis. With probability $1/2^{2d}$, Charlie obtains the state
	\[
	\sum_{a,b,p,q,\in[2^d]}\rho_{a,b,p,q} U_x\pbra{\ket{a}\bra{p}}\otimes  \ket{0^d}\bra{0^d}\otimes \ket{0^d}\bra{0^d}\otimes V_y\pbra{\ket{b}\bra{q}}.
	\]
	As before, Charlie can tell when he obtained this state. Ignoring the zero qubits from $d+1$ to $3d$, this state is  precisely the state Charlie had received in the original $\Qent$ protocol as in Eq.~\eqref{eq:statereceivedbyQ||star}. This completes the proof of \Cref{lemma:simulation_q}.
\end{proof}

\begin{proof}[Proof of \Cref{lemma:simulation_r}]
	
	Let $\Hcal_A,\Hcal_B$ be Hilbert spaces of dimension $2^d$. Let $\Hcal=\Hcal_A\otimes \Hcal_B$ and let $\rho\in\Scal(\Hcal)$ be the state shared by Alice and Bob. We can assume without loss of generality that an $\Roneent$ protocol using $\rho$ as the entanglement has the following form. Suppose Alice and Bob got $x,y$ respectively. Alice first measures her half of the shared state using a $2^c$-outcome measurement operator $\{M_z(x):z\in\bin^c\}$. She obtains an outcome $z\in \bin^c$ and sends $z$ to Bob. Based on this message $z$ and his input $y$, Bob applies the two-outcome POVM $\{N_1(y,z),N_{-1}(y,z)\}$ on his part of the shared state  and outputs the measurement outcome as the answer. Let 
	$\sigma(x,z)$ be the post-measurement state of $\rho$ after Alice applies her measurement $M_z(x)$ and obtains the outcome~$z$. The expected output of Bob is precisely
	\[ \Tr\pbra{(\mathbb{I}\otimes F(y,z))\cdot  \sigma(x,z)}\]
	where $F(y,z)=N^\dagger _1(y,z)N_1(y,z)-N^\dagger _{-1}(y,z)N_{-1}(y,z)$ and $\mathbb{I}\otimes F(y,z)$ is an operator that acts as identity on the first $d$ qubits and acts as $F(y,z)$ on the last $d$ qubits.

	
	Consider an $\Rone$ protocol where Alice obtains the measurement outcome $z\in\bin^c$ and the post-measurement state $\sigma(x,z)\in\Scal(\Hcal)$. She knows the precise classical description of the state $\sigma(x,z)$. She will send the classical message $z$ as before. She samples uniformly random $i,j\in [2^{2d}]$ and sends $(i,j,\widetilde{\sigma}(x,z)[i,j])$, where, $\widetilde{\sigma}(x,z)[i,j]$ is the $(i,j)$th coordinate of $\sigma(x,z)$ specified up to $\Theta(d)$ bits of precision. We will show that using this message, Bob can estimate $\Tr((\mathbb{I}\otimes F(y,z))\cdot \sigma(x,z))$ with advantage $2^{-4d}$. Let $F'(y,z)=\mathbb{I}\otimes F(y,z)$. First, consider an ideal situation where Alice sends exactly $\sigma(x,z)[i,j]$. Observe that
	\begin{align}\label{eq:expectation}
		\begin{split}
			\forall i,j\in[2^{2d}],\quad \abs{F'(y,z)[i,j]\cdot \sigma(x,z)[i,j]}&\le 1,  \\
			\mathop{\E}_{i,j\sim [2^{2d}]}\sbra{ F'(y,z)[i,j]\cdot \sigma(x,z)[i,j] }&=\tfrac{1}{2^{4d}} \Tr(F'(y,z)\cdot  \sigma(x,z)).\end{split}\end{align}
	Consider the protocol where Bob computes $F'(y,z)[i,j]\cdot \sigma(x,z)[i,j]$ for uniformly randomly $i,j\sim[2^{2d}]$. He then returns a random $\bin$-bit whose expectation is $F'(y,z)[i,j]\cdot \sigma(x,z)[i,j]$. This is well-defined due to the first line in Eq.~\eqref{eq:expectation}. The assumption is that $\Tr(F'(y,z)\cdot\sigma(x,z))$ is at least $1/3$ for $\textsc{no}$ instances and at most $-1/3$ for $\textsc{yes}$ instances. This, along with the second line of Eq.~\eqref{eq:expectation} implies that Bob's expected output is at least $\tfrac{1}{3}\cdot 2^{-4d}$ for $\textsc{no}$ instances and at most $-\frac{1}{3}\cdot 2^{-4d}$ for $\textsc{yes}$ instances. Thus, Bob solves the problem with advantage $\tfrac{1}{3}\cdot 2^{-4d}$. Suppose Alice specifies $\widetilde{\sigma}(x,z)[i,j]\in[-1,1]$ up to $5d$ bits of precision, then we have
	\begin{align*} &\abs{\underset{i,j\sim 2^{[2d]}}{\E} \sbra{ F'(y,z)[i,j]\cdot \widetilde{\sigma}(x,z)[i,j]} - 2^{-4d}\cdot \Tr\pbra{F'(y,z)\cdot \sigma(x,z)}}\\
		&\le 2^{-4d}\cdot \abs{\Tr\pbra{F'(y,z)\cdot (\sigma(x,z)-\widetilde{\sigma}(x,z))}}\\
		&\le 2^{-4d}\cdot 2^{-5d}\cdot 2^{4d}\le 2^{-5d}\ll \frac{1}{6}\cdot 2^{-4d}.  \end{align*}
	Following the same calculations as before, it follows that Bob solves the problem with advantage at least $\tfrac{1}{6}\cdot 2^{-4d}$.
\end{proof}

\subsection{XOR Lemma for \texorpdfstring{$\Qpub$}{Q||} for the Boolean Hidden Matching Problem}
\label{sec:five_lower_bound}
Our main technical result in this section is an XOR lemma for $\Qpub$ protocols with regards to the Boolean Hidden Matching problem.

\lemmaQpriXOR*
To prove this lemma, we will make use of some properties which are very similar to those proved in~\cite{gavinsky2007exponential}. The proofs of the facts are deferred to the appendix.  Let $\Mcal$ be the uniform distribution on matchings on $[n]$ of size $m=\alpha  n$ and $\Ucal$ be the uniform distribution on $\{-1,1\}^n$. 

\begin{definition}[$M$ matches $S$]
	Let $S\subseteq[nk]$ and $M\in \supp(\Mcal^{\otimes k})$. We say that $M$ matches $S$ if $M$ is an induced perfect matching on  $S$. If $M$ matches $S$, we use $M(S)\subseteq[mk]$ to denote the subset of edges of this induced perfect matching. 
\end{definition}

Observe that the map $T=M(S)$ defines a bijection between sets $S$ that are matched by $M$ and subsets $T\subseteq[mk]$. Furthermore, $|T|=|S|/2$ and for any $i\in[k]$, $|T_i|$ is odd if and only if $|S_i|/2$ is odd. We now define some sets that will be important in the proof. 
\begin{definition}\label{definition:levels}
	Let $\Scal_{n,k}:=\{S\subseteq [nk]: \forall i\in[k],|S_i|/{2}\text{ is an odd integer}\}$ and $\Tcal_{n,k}:=\{T\subseteq [mk]: \forall i\in[k],|T_i|\text{ is an odd integer}\}$. Define $\Scal^\ell_{n,k}:=\{S\in \Scal_{n,k}:|S|=2\ell\}$ and $\Tcal^\ell_{n,k}:=\{T\in \Tcal_{n,k}:|T|=\ell\}$ for all $\ell\in[mk]$. 
\end{definition}
The aforementioned map $T=M(S)$ provides a bijection between sets $S\in \Scal^{\ell}_{n,k}$ that are matched by $M$ and sets $T\in \Tcal^{\ell}_{n,k}$. The following facts can be proved using techniques in~\cite{gavinsky2007exponential}. 
\begin{fact}\label{fact:matches} Let $S\subseteq[nk]$ and $M\in \supp(\Mcal^{\otimes k})$. Then, for any $w\in\bin^{mk}$, the quantity
	\[ \underset{x\sim\Ucal^{\otimes k} }{\E}\sbra{\indi[Mx=w]\cdot \chi_S(x) } \]
	is non-zero if and only if $M$ matches $S$. Furthermore, if it is non-zero, it equals $2^{-mk}\cdot\chi_{M(S)}(w)$.
\end{fact}

\begin{fact}\label{fact:matchingprobability} Let $S\subseteq[nk]$ with $|S|=2\ell$. Then,
	\[\Pr_{M\sim \Mcal^{\otimes k}}[M\text{ matches }S]\le  O_{\ell}\pbra{\frac{\ell^\ell}{(nk)^\ell}} .\] 
\end{fact}

We now prove our main lemma of this subsection.

\begin{proof}[Proof of \Cref{lemma:main_lemma_3}]
	We assume that $(c/ k)^{3/2}\le \tau \cdot n^{1/2}$ for some small constant $\tau>0$, otherwise the statement of the lemma is vacuously true. We will construct distributions on the \textsc{yes} and \textsc{no} instances of the $\oplus^k$-Boolean Hidden Matching problem such that no small cost $\Qpri$  protocol can distinguish them with considerable advantage. Consider the following two distributions.
	\begin{itemize}
		\item  $\Ncal$ is a distribution on \textsc{no}-instances of $\bhm$:  A random sample of $\Ncal$ is of the form $(x,M,y)$ where $x\sim \Ucal$, $M\sim \Mcal$ and  $y:=Mx$.
		\item  $\Ycal$ is a distribution on \textsc{yes}-instances of $\bhm$ defined similarly to $\Ncal$ except that $y:=\overline{Mx}$.
	\end{itemize}
	Define two distribution $\mu_{1}^{(k)},\mu_{-1}^{(k)}$ on inputs to the $\oplus^k$-Boolean Hidden Matching problem as~follows.
	\begin{equation}\label{eq:definition_mixed_q} 
		\mu_{1}^{(k)}:= \frac{1}{2^{k-1}}\sum_{\substack{K\subseteq[k]\\|K|\text{ is even}}} \Ycal_K \Ncal_{\overline K} \quad \text{and}\quad  \mu_{-1}^{(k)}:= \frac{1}{2^{k-1}}\sum_{\substack{K\subseteq [k]\\|K|\text{ is odd}}} \Ycal_K \Ncal_{\overline K}. 
	\end{equation}
	Recall that $\Ycal_K \Ncal_{\overline K}$ is a product of $k$ independent distributions, where the $i$-th distribution is $\Ycal$ if $i\in K$ and is $\Ncal$ if $i\notin K$. Clearly $\mu_{-1}^{(k)}$ and $\mu_{1}^{(k)}$ are distributions on the \textsc{yes} and \textsc{no} instances respectively of the $\oplus^k$-Boolean Hidden Matching problem. 
	
	Consider any $\Qpri$ protocol with $c$ qubits of communication and let $\Hcal$ be a Hilbert space of dimension $2^c$. Such a protocol can be described by density matrices $\rho(x)\in \Scal(\Hcal)$ and $\sigma_M(y)\in \Scal(\Hcal)$ for every $x\in \{-1,1\}^{nk},y\in \{-1,1\}^{mk}$ and  $M\in\supp(\Mcal^{\otimes k})$. The state received by Charlie on these inputs is precisely $\rho(x)\otimes \sigma_M(y)$. We will show that the trace distance between the states $\mathop{\E}_{(x,M,y)\sim \mu_{1}^{(k)}}\sbra{\rho(x) \otimes \sigma_{M}(y)}$ and $\mathop{\E}_{(x,M,y)\sim \mu_{-1}^{(k)}} \sbra{\rho(x) \otimes \sigma_M(y) }$ is at most  $O_k\pbra{\tfrac{(c/k)^{3k/2}}{n^{k/2}}}$. Since the trace distance measures the maximal distinguishing probability between the two states, this, along with Yao's principle would complete the proof.
	Towards this, define
	\[
	\Delta :=\mathop{\E}_{(x,M,y)\sim \mu_1^{(k)}}\sbra{\rho(x) \otimes \sigma_M(y)} - \mathop{\E}_{(x,M,y)\sim \mu_{-1}^{(k)}} \sbra{\rho(x) \otimes \sigma_M(y)}.
	\]  
	Using the definition of $\mu_1^{(k)}$ and $\mu_{-1}^{(k)}$ in Eq.~\eqref{eq:definition_mixed_q}, we have
	\[
	\Delta\triangleq \sum_{K\subseteq[k]} \frac{(-1)^{|K|}}{2^{k-1}}\cdot \underset{\substack{x\sim \Ucal^{\otimes k}\\ M\sim \Mcal^{\otimes k}}}{\E}\sbra{\rho(x) \otimes \sigma_{M}\pbra{\overline{(Mx)}_K (Mx)_{\overline{K}}}} .
	\]
	We introduce a variable $w\in \{-1,1\}^{mk}$ to represent $Mx$ so that 
	\[ \Delta= \sum_{\substack{w\in \{-1,1\}^{mk}\\ K\subseteq[k]}}\frac{(-1)^{|K|}}{2^{k-1}}\cdot  \underset{\substack{x\sim \Ucal^{\otimes k}\\ M\sim \Mcal^{\otimes k}}}{\E}\sbra{\rho(x) \otimes \sigma_{M}\pbra{\overline{w}_Kw_{\overline{K}}}\cdot \indi[Mx=w]}.\]
	We expand $\rho(x)$ in the Fourier Basis to obtain
	\begin{equation*}	\Delta=\sum_{\substack{w\in \{-1,1\}^{mk}\\ K\subseteq[k]\\S\subseteq [nk]}} \frac{(-1)^{|K|}}{2^{k-1}}\cdot \underset{\substack{x\sim \Ucal^{\otimes k}\\ 
				M\sim \Mcal^{\otimes k}}}{\E}\sbra{\widehat{\rho}(S)\otimes  \sigma_{M}\pbra{\overline{w}_K w_{\overline{K}}}\cdot \sbra{\indi[Mx=w]\cdot \chi_S(x) }}.
	\end{equation*}
	
	Consider the term $\underset{x\sim \Ucal^{\otimes k}}{\E}\sbra{\indi[Mx=w]\cdot \chi_S(x) }$. By~\Cref{fact:matches}, this term is non-zero if and only if $M$ matches $S$, in which case the term evaluates to $2^{-mk}\cdot \chi_{M(S)}(w)$. Substituting this in the equation above, we have that $\Delta$ equals
	\begin{align}\label{eq:fourier_basis_w}
	 \sum_{\substack{w\in \{-1,1\}^{mk}\\ K\subseteq[k]\\S\subseteq [nk]}} \frac{(-1)^{|K|}}{2^{k-1}}    \underset{M\sim\Mcal^{\otimes k}}{\E} \sbra{ \widehat{\rho}(S) \otimes \sigma_{M}\pbra{\overline{w}_K w_{\overline{K}}}\cdot 2^{-mk}\cdot \chi_{M(S)}(w)\cdot \indi[M\text{ matches }S]}.
	\end{align}
	We now expand $\sigma_{M}\pbra{\overline{w}_K w_{\overline K}}$ in the Fourier basis with respect to $w$. Consider
	\begin{align*}
		\sum_{K\subseteq [k]} (-1)^{|K|}\cdot \sigma_{M}\pbra{\overline{w}_K w_{\overline{K}}}
		&=  \underset{ K\subseteq [k] }{\sum}(-1)^{|K|}\cdot \sum_{T\subseteq [mk]}  \widehat{\sigma_{M}}(T)\cdot \chi_T(\overline{w}_K,w_{\overline{K}})\\
		&=   \underset{\substack{K\subseteq [k]\\T\subseteq [mk]}}{\sum}(-1)^{|K|}\cdot \widehat{\sigma_{M}}(T)\cdot\chi_T(w)\cdot (-1)^{\sum_{i\in K}|T_i|} \\
		&=  \underset{T\subseteq [mk]}{\sum} \widehat{\sigma_{M}}(T)\cdot \chi_T(w)\cdot   \sum_{K\subseteq[k]} \big[(-1)^{|K|+\sum_{i\in K}|T_i|}\big].
	\end{align*}
	For $i\in[k]$, let $t_i=-1$ if $|T_i|$ is odd and $t_i=1$ if $|T_i|$ is even. Observe that
	\[ \mathop{\E}_{K\subseteq [k]} \big[(-1)^{|{K}|+\sum_{i\in{K}}|T_i|}\big]=\mathop{\E}_{K\subseteq[k]} \big[\chi_{K}(-t_1,\ldots,-t_k)\big]=\begin{cases} 1&\text{ if }\forall i\in[k],t_i=-1, \\ 0&\text{otherwise.}\end{cases} \]
	Hence, the quantity ${\sum}_{K\subseteq [k]} \big[(-1)^{|{K}|+\sum_{i\in{K}}|T_i|}\big]$ is non-zero if and only if $|T_i|$ is odd for all $i\in[k]$. Furthermore, if it is non-zero, then it equals $2^{k}$. Recall that we defined $\Tcal_{nk}:=\{T\subseteq[mk]:\forall i\in[k],|T_i|\text{ is odd}\}$ in~\Cref{definition:levels}. Using this, we have
	\begin{align}
		\label{eq:sigmatofouriersigma}
		\sum_{K\subseteq [k]} (-1)^{|K|}\cdot \sigma_{M}\pbra{\overline{w}_Kw_{\overline{K}}} = 2^k\cdot  \underset{T\in \Tcal_{n,k}}{\sum}\widehat{\sigma_M}(T)\cdot \chi_T(w).
	\end{align}
	Substituting this in Eq.~\eqref{eq:fourier_basis_w}, we have that $\Delta$ equals
	\begin{align*}
	&2 \underset{\substack{w\in\{-1,1\}^{mk}\\ S\subseteq [nk]}}{\sum}  \underset{M\sim \Mcal^{\otimes k}}{\E} \sbra{\widehat{\rho}(S)\otimes \underset{T\in \Tcal_{n,k}}{\sum} \widehat{\sigma_M}(T)\cdot 2^{-mk}\cdot \chi_{M(S)}(w)\cdot \chi_T(w) \cdot \indi[M\text{ matches }S]}\\
		&=2 \underset{\substack{ S\subseteq [nk]}}{\sum}  \underset{M\sim \Mcal^{\otimes k}}{\E} \sbra{\widehat{\rho}(S)\otimes \underset{T\in \Tcal_{n,k}}{\sum} \widehat{\sigma_M}(T)\cdot \mathop{\E}_{w\sim \pmset{mk}} [\chi_{M(S)+T}(w)]\cdot \indi[M\text{ matches }S]}.
	\end{align*}
	Observe that if $M$ matches $S$, then $\E_{w\sim\{-1,1\}^{mk}}\sbra{\chi_{M(S)+T}(w)}$ equals $1$ if $T=M(S)$ and $0$ otherwise. Recall that the sets $S\subseteq[nk]$ such that $M$ matches $S$ and $M(S)\in \Tcal_{n,k}$ are precisely those sets in $\Scal_{n,k}$ that are matched by $M$. Hence,
	\[	
	\Delta
	= \underset{S\in \Scal_{n,k}}{\sum} \widehat{\rho}(S) \otimes \E_{M} \sbra{ \widehat{\sigma_{M}}(M(S)) \cdot \indi[M\text{ matches }S]}.
	\]
	We now upper bound $\|\Delta\|_1$ by triangle inequality as follows. 
	\[
	\vabs{\Delta}_1
	\le \underset{S\in \Scal_{n,k}}{\sum} \vabs{\widehat{\rho}(S)}_1 \otimes \E_{M} \sbra{ \vabs{\widehat{\sigma_{M}}(M(S))}_1 \cdot \indi[M\text{ matches }S]}.
	\]
	We partition $\Scal_{n,k}$ and $\Tcal_{n,k}$ into $\sqcup_\ell\Scal^{\ell}_{nk}$ and $\sqcup_\ell \Tcal^{\ell}_{n,k}$ based on the size of the sets as in~\Cref{definition:levels}. Observe that every set in $\Scal_{n,k}$ has size at least $2k$ and every set in $\Tcal_{n,k}$ has size at least $k$. Thus, 
	\begin{equation}\label{eq:partition_level}
		\vabs{\Delta}_1\le \sum_{\ell= k}^{mk} \underset{S\in \Scal^\ell_{k,n}}{\sum} \vabs{\widehat{\rho}(S)}_1 \otimes \E_{M} \big[ \vabs{\widehat{\sigma_{M}}(M(S))}_1 \cdot 1[M\text{ matches }S]\big]. 
	\end{equation}
	We now apply Cauchy-Schwarz to conclude that 
	\begin{align*} &\underset{M\sim \Mcal^{\otimes k}}{\E}  \big[ \vabs{\widehat{\sigma_{M}}(M(S))}_1\cdot \indi[M\text{ matches }S]\big]\\
		&\le \sqrt{ \underset{M\sim \Mcal^{\otimes k}}{\E}  \big[ \vabs{\widehat{\sigma_{M}}(M(S))}_1^2\cdot \indi[M\text{ matches }S]\big]}\cdot \sqrt{\underset{M\sim \Mcal^{\otimes k}}{\Pr}[M\text{ matches S}]}.
	\end{align*}
	\Cref{fact:matchingprobability} implies that for any $S\in \Tcal_{n,k}^{\ell}$, we have $\underset{M\sim\Mcal^{\otimes k}}{\Pr}[M\text{ matches }S]\le O_\ell\pbra{\frac{\ell^\ell}{(nk)^\ell}}$. Substituting this in Eq.~\eqref{eq:partition_level} implies that
	\[	
	\|\Delta\|_1	\le \sum_{\ell=k}^{mk} \underset{S\in \Scal^\ell_{nk}}{\sum}\vabs{ \widehat{\rho}(S)}_1\cdot \sqrt{\E_{M} \sbra{ \vabs{\widehat{\sigma_{M}}(M(S))}_1^2 \cdot \indi[M\text{ matches }S] }}\cdot O_\ell\pbra{\frac{\ell^{\ell/2}}{(nk)^{\ell/2}}}.
	\] 
	Again, by Cauchy-Schwarz, we have
	\[
	\|\Delta\|_1 \le \sum_{\ell=k}^{mk} \sqrt{\underset{S\in \Scal^\ell_{n,k}}{\sum}\vabs{ \widehat{\rho}(S)}_1^2}\cdot \sqrt{\underset{S\in \Scal^\ell_{n,k}}{\sum}  \E_{M} \sbra{ \vabs{\widehat{\sigma_{M}}(M(S)) }_1^2\cdot \indi[M\text{ matches }S}} \cdot  O_\ell\pbra{\frac{\ell^{\ell/2}}{(nk)^{\ell/2}}}.
	\]
	By the aforementioned correspondence between sets $S\in \Scal^\ell_{n,k}$ such that $M$ matches $S$ and sets $T\in \Tcal^\ell_{n,k}$, we have
	\begin{equation}
		\label{eq:matrix_level_k} 
		\|\Delta\|_1 \le \sum_{\ell=k}^{mk} \sqrt{\underset{S\in \Scal^\ell_{n,k}}{\sum}\vabs{ \widehat{\rho}(S)}_1^2}\cdot \sqrt{\underset{T\in \Tcal^\ell_{n,k}}{\sum}  \E_{M} \sbra{ \vabs{\widehat{\sigma_{M}}(T) }_1^2 }} \cdot  O_\ell\pbra{\frac{\ell^{\ell/2}}{(nk)^{\ell/2}}}.
	\end{equation}
	We now apply the Matrix level-$k$ inequality in \Cref{lemma:level_k_inequality_matrix} to the functions $\rho:\{-1,1\}^n\to \Scal(\Hcal)$ and $\sigma_M:\{-1,1\}^m\to \Scal(\Hcal)$ where $\Hcal$ is a Hilbert space of dimension $2^c$. \Cref{lemma:level_k_inequality_matrix} implies that 
	\[ 
	\sum_{|S|=2\ell}\vabs{\widehat{\rho}(S)}_1^2  \le O_\ell\pbra{(c/\ell)^{2\ell}}+O_\ell(1)  \quad\text{and}\quad \sum_{|T|=\ell}\vabs{\widehat{\sigma_M}(T)}_1^2 \le O_\ell\pbra{(c/\ell)^\ell}+O_\ell(1). 
	\]
	Substituting this in Eq.~\eqref{eq:matrix_level_k}, we get
	\begin{align*}
		\vabs{\Delta}_1&\le \sum_{\ell=k}^{mk} \sqrt{\underset{S:|S|=2\ell}{\sum}\vabs{ \widehat{\rho}(S)}_1^2}\cdot \sqrt{\underset{T:|T|=\ell}{\sum}  \E_{M} \sbra{ \vabs{\widehat{\sigma_{M}}(T) }_1^2 }} \cdot  O_\ell\pbra{\frac{\ell^{\ell/2}}{(nk)^{\ell/2}}}\\
		&\le \sum_{\ell=k}^{mk} \max\pbra{ O_\ell\pbra{\frac{c^{3\ell/2}}{\ell^{3\ell/2}}}  , O_\ell(1)}\cdot   O_\ell\pbra{\frac{\ell^{\ell/2}}{(nk)^{\ell/2}}}
		\\
		&\le \sum_{\ell=k}^{mk}  O_\ell\pbra{\frac{c^{3\ell/2}}{\ell^{\ell}(nk)^{\ell/2}}} +\sum_{\ell=k}^{mk}  O_\ell\pbra{\frac{\ell^{\ell/2}}{(nk)^{\ell/2}}}.  
	\end{align*}
	Since $\ell\le mk=\alpha\cdot nk$ for a sufficiently small constant $\alpha>0$, the function $\ell^{\ell/2}/(nk)^{\ell/2}$ is exponentially decaying for $\ell\in[k,mk]$ and hence the second term is at most $O_k\pbra{n^{-k/2}}$. Our assumption that $(c/ k)^{3/2} \le \tau \cdot n^{1/2}$ for a sufficiently small constant $\tau>0$ implies that the function $c^{3\ell/2}/(\ell^\ell(nk)^{\ell/2})$ is exponentially decaying for $\ell\in [k,mk]$ and hence, the first term above is at most $O_k\pbra{\frac{(c/k)^{3k/2}}{n^{k/2}}}$. Together, we have
	\[\|\Delta\|_1\le O_k\pbra{\frac{(c/k)^{3k/2}}{n^{k/2}}} + O_k(n^{-k/2}). \]
	This completes the proof of \Cref{lemma:main_lemma_3}.
\end{proof}

\subsection{XOR Lemma for \texorpdfstring{$\Rone$}{R1} for the Boolean Hidden Matching Problem}

In this subsection, we prove~\Cref{lemma:main_lemma_4} which we restate here for convenience. 
\lemmaRoneXOR*

\begin{proof}[Proof of~\Cref{lemma:main_lemma_4}]
	The proof of this lemma will be similar to the proof of Lemma~\ref{lemma:main_lemma_3} and hence we will follow similar notation. Let $z\in \bin^c$ be any $c$-bit message sent by Alice and let $A_z\subseteq\{-1,1\}^{nk}$ be the set of Alice's inputs for which Alice would have sent $z$ to Bob. Let $g(x)=\indi[x\in A_z]$. Fix any $M\in \supp(\Mcal^{\otimes k})$. Similar to Lemma~\ref{lemma:main_lemma_3}, let $\Ncal^M(y)$ be the distribution on $y\in \{-1,1\}^{mk}$ induced by sampling $x\sim A_z$ and letting $y=Mx$. Let $\Ycal^M(y)$ be similarly defined  with $y:=\overline{Mx}$. So we have that $\Ncal^M(y)=\frac{\abs{\{ x\in A_z| Mx=y\}}}{|A_z|}$ and $\Ycal^M(y)=\frac{\abs{\{ x\in A_z| \overline{Mx}=y\}}}{|A_z|}$ for all $y\in\{-1,1\}^{mk}$. Define 
	\begin{equation}
		\label{eq:definition_mixed_r}
		\mu_1^{(k)}:= \frac{1}{2^{k-1}}\sum_{\substack{K\subseteq[k]\\|K|\text{ is even}}} \Ycal^M_K \Ncal_{\overline K}^M \quad \text{and}\quad  \mu_{-1}^{(k)}:= \frac{1}{2^{k-1}}\sum_{\substack{K\subseteq [k]\\|K|\text{ is odd}}} \Ycal_K \Ncal_{\overline K}^M. 
	\end{equation}
	Below we show that for a typical $M\sim \Mcal^{\otimes k}$, these two distributions are close in total variational~distance. By arguments similar to~\cite{gavinsky2007exponential}, this would complete the proof. To this end,~let
	\[
	\Delta_{A_z}:=\underset{M\sim\Mcal^{\otimes k}}{\E}\sbra{\vabs{ \mu_{1}^{(k)}-\mu_{-1}^{(k)}}_1}.
	\]
	By Eq.~\eqref{eq:definition_mixed_r}, we have $\mu_{1}^{(k)}-\mu_{-1}^{(k)}=2^{1-k}\cdot\sum_{K\subseteq [k]}(-1)^{|K|}\Ycal^M_K\Ncal^M_{\overline{K}}$. Hence
	\begin{align*}
		\Delta_{A_z}^2&\le 2^{mk}\cdot \E_M\sbra{\vabs{ \mu_{1}^{(k)}-\mu_{-1}^{(k)}}_2^2} 
		= 2^{mk} \cdot \E_M\sbra{ 2^{-2k+2}\cdot \vabs{ \sum_{K\subseteq [k]} \Ycal^M_K\Ncal^M_{\overline{K}} (-1)^{|K|}  }_2^2},
	\end{align*}
	where the first inequality is by the Cauchy-Schwarz inequality. By Parseval's theorem, we have 
	\begin{align}\label{eq:eq1}
		\Delta_{A_z}^2&\le 2^{2mk-2k+2} \cdot \E_M\sbra{\sum_{\substack{T\subseteq [mk]\\T\neq \emptyset}}\pbra{ \sum_{K\subseteq [k]} \widehat{\Ycal^M_K\Ncal^M_{\overline{K}}}(T)(-1)^{|K|}  }^2}.
	\end{align}
	Observe that 
	\begin{align*}
		\widehat{\Ycal^M_K\Ncal^M_{\overline{K}}}(T)
		&=\frac{1}{2^{mk}}\sum_{y\in \{-1,1\}^{mk}} \pbra{\Ycal^M_K\Ncal^M_{\overline{K}}}(y)\cdot \chi_T(y)\\
		&= \frac{1}{2^{mk}\cdot |A_z|}\Big( \abs{\cbra{x\in A_z|\chi_T\pbra{(Mx)_{\overline{K}}(\overline{Mx})_{{K}}}=1}}\\&- \abs{\cbra{x\in A_z\mid\chi_T\pbra{(Mx)_{\overline{K}}(\overline{Mx})_{{K}}}=-1}}\Big)\\
		&=  \frac{1}{2^{mk}\cdot |A_z|} \Big(\abs{\cbra{x\in A_z\mid\chi_T(Mx)= (-1)^{\sum_{i\in K}|T_i|} }} \\&-\abs{\cbra{x\in A_z\mid\chi_T(Mx)\neq  (-1)^{\sum_{i\in K}|T_i|}}} \Big)\\	
		&=  \frac{1}{2^{mk}} \sum_{y\in\{-1,1\}^{mk}} \Ncal^{\otimes k}(y)\cdot \chi_T(y)\cdot (-1)^{\sum_{i\in K} |T_i|}\\
		&= \widehat{\Ncal^{\otimes k}}(T)\cdot (-1)^{\sum_{i\in K}|T_i|}.
	\end{align*}
	By an argument analogous to~\cite[Eq.~(3)]{gavinsky2007exponential}, we~have $\widehat{\Ncal^{\otimes k}}(T)=\tfrac{2^{nk}}{|A_z|\cdot 2^{mk}}\cdot \widehat{g}(M^\dagger T)$.
	Hence
	\begin{align}\label{eq:cancel}
		\sum_{K\subseteq [k]} \widehat{\Ycal^M_K \Ncal^M_{\overline K}}(T)\cdot (-1)^{|K|}&=\frac{2^{nk}}{2^{mk}\cdot |A_z|}\cdot\widehat{g}(M^\dagger T)  \cdot \sum_{K\subseteq [k]} (-1)^{|K|+\sum_{i\in K}|T_i|}.
	\end{align}
	As we saw in Eq.~\eqref{eq:sigmatofouriersigma}, the term $\sum_{K\subseteq [k]} (-1)^{|K|+\sum_{i\in  K}|T_i|}$ is $2^k$ if $|T_i|$ is odd for all $i\in[k]$ and zero otherwise. Hence, the R.H.S. of Eq.~\eqref{eq:cancel} is non-zero only if $T\in \Tcal_{n,k}$ (defined in~\Cref{definition:levels}), and in this case equals $\frac{2^{nk}}{2^{mk}\cdot |A_z|}\cdot \widehat{g}(M^\dagger T)\cdot 2^k$. Substituting this in Eq.~\eqref{eq:eq1}, we~have that $	\Delta_{A_z}^2$ equals
	\begin{align*}
	2^{2mk-2k+2}\cdot \E_{M}\sbra{\sum_{T\in \Tcal_{n,k}} \frac{2^{2nk+2k}}{2^{2mk}\cdot |A_z|^2} \widehat{g}(M^\dagger T)^2}=4\cdot \E_{M}\sbra{\sum_{T\in \Tcal_{n,k}} \frac{2^{2n}}{ |A_z|^2} \widehat{g}(M^\dagger T)^2}.
	\end{align*}
	Recall the correspondence between $\Tcal_{n,k}$ and $\Scal_{n,k}$ as in~\Cref{definition:levels}.   For every $S\in \Scal_{n,k}$, there is at most one $T\in \Tcal_{n,k}$ such that $M^\dagger T=S$, furthermore, such a $T$ exists if and only if $M$ matches~$S$. Hence we have that
	\begin{align*}
		\Delta_{A_z}^2&\le 4\cdot \E_{M}\sbra{\sum_{S\in \Scal_{n,k}} \frac{2^{2n}}{ |A_z|^2} \widehat{g}(S)^2 \cdot \indi[M\text{ matches }S]}\\
		&=4\cdot \sum_{\ell= k}^{mk}\sum_{S\in \Scal^{\ell}_{n,k}} \frac{2^{2n}}{ |A_z|^2} \widehat{g}(S)^2 \cdot \Pr_M[M\text{ matches }S]\le \sum_{\ell= k}^{mk} \pbra{\sum_{|S|=2\ell } \frac{2^{2n}}{ |A_z|^2} \widehat{g}(S)^2} \cdot O_\ell\pbra{\frac{\ell^\ell}{(nk)^\ell}},
	\end{align*}
	where we used \Cref{fact:matchingprobability}. Let $\mu(A_z)=\tfrac{|A_z|}{2^n}$. Applying \Cref{lemma:level_k_inequality}, we have
	\begin{align*}
		\Delta_{A_z}^2& \le   \sum_{\ell= k}^{mk} \pbra{2e\cdot \ln\pbra{\tfrac{e}{\mu(A_z)^{1/(2\ell)}}}}^{2\ell}\cdot O_\ell \pbra{ \tfrac{\ell^\ell}{ (nk)^\ell}}.
	\end{align*}
	We now take square root on both sides (and use concavity of the square root function) to get
	\[\Delta_{A_z}\le \sum_{\ell=k}^{mk} \pbra{2e\cdot \ln\pbra{\tfrac{e}{\mu(A_z)^{1/(2\ell)}}}}^{\ell}\cdot O_\ell \pbra{ \tfrac{\ell^{\ell/2}}{ (nk)^{\ell/2}}}. \]
	We now multiply both sides by $\mu(A_z)$ and add over all $2^c$ possibilities for the transcript $z\in\bin^c$.
	\begin{align*}
		\Delta&:=\sum_{z\in\bin^c}\Delta_{A_z} \le  \sum_{z\in\bin^c} \mu(A_z) \cdot \sum_{\ell= k}^{mk} \pbra{2e\cdot \ln\pbra{\tfrac{e}{\mu(A_z)^{1/(2\ell)}}}}^{\ell}\cdot O_\ell \pbra{ \tfrac{\ell^{\ell/2}}{ (nk)^{\ell/2}}}.
	\end{align*}
	We now use the concavity of the function $h'(\gamma)=\gamma\cdot \ln(e/\gamma^{1/2\ell})^\ell$  for $\gamma\in[0,1]$ and all $\ell\in \Nbb$,  (similarly to the proof in~\Cref{sec:fourier_growth_rtwo}) to conclude that
	\[\Delta\le \sum_{\ell=k}^{mk} \pbra{\sum_{z\in\bin^c} \mu(A_z) }\cdot \pbra{2e\cdot \ln\pbra{\frac{e\cdot 2^{c/(2\ell)}}{\pbra{\sum_{z\in\bin^c}\mu(A_z)}^{1/(2\ell)}}  }}^\ell \cdot O_\ell\pbra{\frac{\ell^{\ell/2}}{(nk)^{\ell/2}}}. \]
	We use the fact that $\sum_{z\in\bin^c}\mu(A_z)=1$ to conclude that 
	\begin{align*}
		\Delta&\le \sum_{\ell=k}^{mk}  \pbra{2e\cdot \ln\pbra{e\cdot 2^{c/(2\ell)}  }}^\ell \cdot O_\ell\pbra{\frac{\ell^{\ell/2}}{(nk)^{\ell/2}}} \le  \sum_{\ell=k}^{mk}  O_\ell\pbra{\frac{\ell^{\ell/2}}{(nk)^{\ell/2}}} +\sum_{\ell=k}^{mk} O_\ell\pbra{\frac{c^{\ell}}{(\ell nk)^{\ell/2}}}.
	\end{align*}
	As before, the first term is at most $O(n^{-k/2})$. The assumption that $c\cdot k\le \tau\cdot n^{1/2}$ for a small enough constant $\tau>0$ implies that the function $\frac{c^{\ell}}{(\ell nk)^{\ell/2}}$ is exponentially decaying for $\ell\in[k,mk]$. Hence, the second term is at most $O_\ell\pbra{\frac{(c/k)^k}{n^{k/2}}}$. This, along with the techniques of~\cite{gavinsky2007exponential} completes the proof of~\Cref{lemma:main_lemma_4}.
\end{proof}

\section{Hybrid Quantum-Classical Lower Bounds}

We begin by defining the hybrid models that we consider. We then present our variant of the matrix level-$k$ inequality and finally prove the hybrid quantum-classical separations. 
\subsection{Hybrid Models}

We consider two hybrid models, namely $(\Rtwo+\Qpri)$ and $(\Rone+\Qpri)$. In the former model, Alice and Bob engage in classical interactive communication, and then each send a quantum message to Charlie, who applies an arbitrary projective measurement and returns the outcome as the answer. In the latter model, Alice sends Bob a classical message and then Alice and Bob each send a quantum message to Charlie, who applies an arbitrary projective measurement and returns the outcome as the answer. Note that in both these models, the quantum state sent by the players to Charlie is allowed to depend on the classical transcript of the protocol. We measure the communication cost by the sum of the quantum and classical cost. 

The main contributions of this section are the proof of~\Cref{theorem:hybrid_1} and~\Cref{theorem:hybrid_2}.

\paragraph*{High-level Ideas.} The $k=1$ version of~\Cref{theoremRtwo} and~\Cref{theoremQparallel} prove that the Forrelation problem requires $\Omega(n^{1/4})$ communication in the $\Rtwo$ and $\Qpub$ models respectively. The first theorem uses classical hypercontractivity and the second one uses matrix hypercontractivity. To prove a combined $(\Rtwo+\Qpub)$ lower bound, the idea is to combine both the hypercontractive inequalities and give a more general version of matrix hypercontractivity encapsulating both. And indeed, using this variant, one can combine the techniques of~\Cref{theoremRtwo} and~\Cref{theoremQparallel} and show the desired bound. The same idea also works to give the $(\Rone+\Qpub)$ lower bound for the Boolean Hidden Matching Problem. We describe these in more detail in the following sections.  

\subsection{A Variant of Matrix Hypercontractivity}
\begin{lemma} 
    Let $\Hcal$ be a Hilbert space of dimension $2^c$ and let $f:\{-1,1\}^n\to \Pcal(\Hcal)$ be a matrix-valued function such that with probability $\alpha$, $f$ is density-matrix valued and with probability $1-\alpha$, $f$ is the all-zeroes matrix. Then, for any $\ell\in \Nbb$ such that $\ell \le (2\ln 2)c+2\ln(1/\alpha)$,
    \[ \sum_{|S|=\ell} \Tr^2(\abs{\widehat{\rho}(S)})\le \alpha^2\cdot (c+\ln(1/\alpha))^\ell \cdot (2e^2\ln 2/\ell)^\ell \] \label{lemma:level_k_inequality_matrix_probability_1}
\end{lemma}

This variant essentially combines the classical level-$k$ inequality (\Cref{lemma:level_k_inequality}) with the matrix-valued one (\Cref{lemma:level_k_inequality_matrix_1}). In more detail, if we set $c=0$, we obtain the classical level-$k$ inequality and if we set $\alpha=1$, we obtain the matrix level-$k$ inequality. The proof of this is identical to the one in~\cite{ben2008hypercontractive}, with the only difference being the setting of parameters. We now provide the details of this proof. 

For a matrix $M$ of dimensions $2^c\times 2^c$, let
\[\|M\|_{\ell_p}=\pbra{2^{-c}\cdot  \Tr(\|M^p\|)}^{1/p}=\pbra{2^{-c}\cdot \sum_i \sigma_i(M)^p}^{1/p}\]
denote the $p$-th norm of its singular values normalized by the dimension. 
\begin{proof}
    Let $p=1+\delta$ where $\delta= \tfrac{\ell}{(2\ln 2) c+2 \ln (1/\alpha)}$. By the assumption on $\ell$, we have $p\le 2$. From~\cite[Theorem 1]{ben2008hypercontractive}, we have
\[
\sum_{S\subseteq[n]} \delta^{|S|} \vabs{\widehat{\rho}(S)}^2_{\ell_p} \le \pbra{ \E\sbra{\vabs{f(x)}_{\ell_p}^p}}^{2/p}
\le \pbra{2^{-c}\cdot \alpha}^{2/p}.\] 
We now use the fact the the $p$-norms are increasing to conclude that
\[
\sum_{S\subseteq[n]} \delta^{|S|} \Tr\pbra{\widehat{\rho}(S)}^2 \le  2^{2c}\cdot \pbra{2^{-c}\cdot \alpha}^{2/p}=2^{2c(1-1/p)}\cdot \alpha^{2/p} \le 2^{2c(p-1)}\cdot \alpha^{2/p}.\]
Firstly, since $p-1\le \frac{\ell}{(2\ln 2) c}$, we have $2^{2c(p-1)}\le e^{\ell}$. Secondly, since $1/p\ge 1-\tfrac{\ell}{2\ln(1/\alpha)}$ and $\alpha\le 1$, we have $\alpha^{2/p}\le \alpha^2 \cdot \alpha^{-\ell/\ln(1/\alpha)}=\alpha^2\cdot e^{\ell}$. We now restrict our attention to terms $S$ such that $|S|=\ell$. In this case, we have $\delta^{-|S|}=\pbra{\frac{(2\ln 2)c+2\ln(1/\alpha)}{\ell}}^{\ell}$. Putting this together, we have
\[\sum_{|S|=\ell} \Tr(\widehat{\rho}(S))^2 \le \alpha^2 \cdot \pbra{\frac{(2\ln 2)c+2\ln(1/\alpha)}{\ell}}^{\ell}\cdot e^{2\ell}\le \alpha^2\cdot (c+\ln(1/\alpha))^\ell \cdot (2e^2\ln 2/\ell)^\ell. \]
\end{proof}

As in the proof of~\Cref{lemma:level_k_inequality_matrix} from~\Cref{lemma:level_k_inequality_matrix_1}, we have the following. 
\begin{corollary} \label{lemma:level_k_inequality_matrix_probability} Under the same hypothesis as~\Cref{lemma:level_k_inequality_matrix_probability_1}, for all $\ell\in \Nbb$, we have
\[\sum_{|S|=\ell} \Tr(\widehat{\rho}(S))^2 \le \alpha^2\cdot\pbra{ O_{\ell}\pbra{\pbra{\tfrac{c+\ln(1/\alpha)}{\ell}}^\ell} + O_\ell(1) }. \]
\end{corollary}

\subsection{Proof of \texorpdfstring{\Cref{theorem:hybrid_1}}{Theorem 6.1}}
As in the proof of~\Cref{theoremRtwo}, it suffices to prove a Fourier growth bound on the XOR-Fiber.
\begin{lemma}\label{lemma:hybrid_forrelation_fourier} Let $C:\{-1,1\}^n\times\{-1,1\}^n\to [-1,1]$ be a $(\Rtwo+\Qpri)$-protocol of classical cost $c$ and quantum cost $q$. Let $H$ be the XOR-fiber of $C$ as in~\Cref{definition:XOR_fiber}. Then, 
\[ L_{1,\ell}(H)\triangleq \sum_{|S|=\ell}\abs{\widehat{H}(S)}\le O_\ell(q^\ell \cdot c^\ell).\]
\end{lemma}
Similarly to the calculation in~\Cref{sec:firsttheoremproof}, it would follow that any $(\Rtwo+\Qpri)$-protocol for the Forrelation problem must satisfy $q^2\cdot c^2\ge \Omega(\sqrt{n})$ and hence $q+c\ge \Omega(n^{1/8})$. This describes the proof of~\Cref{theorem:hybrid_1} from~\Cref{lemma:hybrid_forrelation_fourier}. We now prove~\Cref{lemma:hybrid_forrelation_fourier}. 
\begin{proof}[Proof of~\Cref{lemma:hybrid_forrelation_fourier}]
    
Fix any $(\Rtwo+\Qpri)$ protocol for the Forrelation Problem of classical and quantum communication cost $c$ and $q$ respectively. This induces a family of matrix-valued functions $\{ \rho_z(x)\in \Pcal(\Hcal) \}_{z\in \{-1,1\}^c}$ and $\{ \sigma_z(y) \}_{z\in \{-1,1\}^c}\in \Pcal(\Hcal)$ where $\Hcal$ is a Hilbert space of dimension $2^q$. Here, $\rho_z(x)$ is the zero matrix if $x$ is not compatible with the transcript $z$ and is otherwise the message sent by Alice to Charlie. Similarly, $\sigma_z(y)$ is the zero matrix if $y$ is not compatible with $z$ and is otherwise the message sent by Bob to Charlie. In particular, $\rho_z(x)\otimes \sigma_z(y)$ is the density matrix received by Charlie if the transcript on inputs $x,y$ is $z$ and is the zeroes matrix otherwise. The acceptance probability of the protocol can be described as 
\[ \sum_{z\in\{-1,1\}^c} \Tr\pbra{E\cdot (\rho_z(x)\otimes \sigma_z(y))}.\]
By a calculation similar to the one in~\Cref{sec:secondthmproof}, we have 
\begin{align}\label{eq:hybrid_1}\begin{split}L_{1,2}(H)&=\sum_{|S|=\ell}\abs{\sum_{z\in\{-1,1\}^c}\Tr\pbra{E\cdot (\widehat{\rho_z}(S)\otimes \widehat{\sigma_z}(S))}}\\
&\le \sum_{z\in\{-1,1\}^c}\sum_{|S|=\ell}\Tr\pbra{\abs{\widehat{\rho_z}(S)}}\cdot  \Tr\pbra{\abs{\widehat{\sigma_z}(S)}}\\
&\le \sum_{z\in\{-1,1\}^c}\sqrt{\sum_{|S|=\ell}\Tr^2\pbra{\abs{\widehat{\rho_z}(S)}}}\cdot  \sqrt{\sum_{|S|=\ell}\Tr^2\pbra{\abs{\widehat{\sigma_z}(S)}}}\end{split}\end{align}
where the last inequality follows by Cauchy-Schwarz. Let $\alpha_z=\E_x[\Tr(\rho_z(x))]$ and $\beta_z=\E_y[\Tr(\sigma_z(y))]$ be the probability that $\rho_z$ and $\sigma_z$ are non-zero respectively. We now apply~\Cref{lemma:level_k_inequality_matrix_probability} to the matrix-valued functions $\alpha_z,\beta_z$ to conclude that
\[ \sum_{|S|=\ell}\Tr^2\pbra{\abs{\widehat{\rho_z}(S)}} \le \alpha_z^2\cdot O_\ell\pbra{(q+\ln(1/\alpha_z))^\ell},\]
\[ \sum_{|S|=\ell}\Tr^2\pbra{\abs{\widehat{\sigma_z}(S)}} \le \beta_z^2\cdot O_\ell\pbra{(q+\ln(1/\beta))^\ell}.\]
Substituting this in~\Cref{eq:hybrid_1}, we have 
\begin{equation}\label{eq:hybrid_2} L_{1,\ell}(H)\le \sum_{z\in\{-1,1\}^c}\alpha_z\beta_z\cdot O_\ell\pbra{q^{\ell}\cdot \ln(1/\alpha_z)^{\ell/2}\cdot \ln(1/\beta_z))^{\ell/2}}.\end{equation}
At this point, we need to upper bound $\sum_{z\in\{-1,1\}^c} \alpha_z\beta_z \cdot \ln(1/\alpha_z)^{\ell/2}\cdot \ln(1/\beta_z)^{\ell/2}$ where $\alpha_z,\beta_z\in\{0,1\}$ satisfy $\sum_{z\in\{-1,1\}^c} \alpha_z\beta_z=1$. Consider the simple case when each $\alpha_z=\beta_z=2^{-c/2}$. In this case, the quantity we wish to bound is $O_{\ell}(c^{\ell})$. It turns out that this is essentially the worst case. This requires some convexity arguments and was done explicitly in~\cite{DBLP:conf/approx/GirishRZ21}. We use their results to conclude that  $\sum_{z\in\{-1,1\}^c} \alpha_z\beta_z \cdot \ln(1/\alpha_z)^\ell\cdot \ln(1/\beta_z)^\ell \le O_\ell(c^{\ell})$. This, along with~\Cref{eq:hybrid_2} implies that
\[L_{1,\ell}(H)\le O_{\ell}(q^\ell \cdot c^{\ell}).\]
\end{proof}

\subsection{Proof of \texorpdfstring{\Cref{theorem:hybrid_2}}{Theorem 6.2}}

The proof of this is quite similar to the proof of~\Cref{lemma:main_lemma_3} for the case $k=1$. Fix any classical message $z\in\{-1,1\}^c$ sent by Alice. Define matrix-valued functions $\rho_z(x),\sigma_{z,M}(y)$ which represents the states sent by Alice and Bob given that the transcript is $z$ and Bob's input matching is $M$. In more detail, if the transcript $z$ is incompatible with $x$, then $\rho_z(x)=0$, otherwise, $\rho_z(x)$ and $\sigma_{z,M}(y)$ are the states sent by Alice and Bob. Let $\mu_1=\mu_1^{(1)},\mu_{-1}^{(1)}$ be distributions on \textsc{yes} and \textsc{no} instances defined in~\Cref{eq:definition_mixed_q}. Let
\[\Delta:=\E_{(x,M,y)\sim \mu_1}\sbra{\sum_{z\in\{-1,1\}^c}\rho_z(x)\otimes \sigma_{M,z}(y)}-\E_{(x,M,y)\sim \mu_1}\sbra{\sum_{z\in\{-1,1\}^c}\rho_z(x)\otimes \sigma_{M,z}(y)}  \]
As in the proof of~\Cref{lemma:main_lemma_3}, we will show that the trace norm of $\Delta$ is at most $O((cq)^{3/2}/\sqrt{n})$. We perform a calculation similar to before and obtain the following analogue of~\Cref{eq:matrix_level_k}.
\begin{equation}\label{eq:hybrid_3} \vabs{\Delta}_1 \le \sum_{z\in\{-1,1\}^c}\sum_{\ell \in\Nbb}\sqrt{\sum_{|S|=2\ell} \|\widehat{\rho_z}(S)\|_1^2}\cdot \sqrt{\sum_{|T|=\ell}\E_M\sbra{\|\widehat{\sigma_{z,M}}(T)\|_1^2 }}\cdot O_\ell\pbra{\frac{\ell^{\ell/2}}{(nk)^{\ell/2}}}. \end{equation}
Let $\alpha_z=\E_x[\Tr(\rho_z(x))]$ be the probability with which the transcript is $z$. We apply~\Cref{lemma:level_k_inequality_matrix_probability} to the matrix-valued function $\rho_z$ and~\Cref{lemma:level_k_inequality_matrix} to the density-matrix-valued function $\sigma_z$ to conclude that for all $M$,
\[ \sum_{|S|=2\ell}\|\widehat{\rho_z}(S)\|_1^2\le \alpha_z^2\cdot \max\pbra{ O_\ell\pbra{\frac{(q+\ln(1/\alpha))^{2\ell} }{\ell^{2\ell}}} ,O_\ell(1)} ,\]
\[ \sum_{|T|=\ell}\|\widehat{\sigma_{M,z}}(S)\|_1^2\le \max\pbra{ O_\ell\pbra{\frac{q^\ell }{\ell^\ell}} ,O_\ell(1)} ,\]
Substituting this in~\Cref{eq:hybrid_3}, we get
\begin{equation}\label{eq:hybrid_4} \| \Delta \|_1 \le \sum_{z\in\{-1,1\}^c}\sum_{\ell\in \Nbb} \alpha_z \cdot \max\pbra{O_\ell\pbra{\frac{ q^{3\ell/2}\ln(1/\alpha_z)^{3\ell/2}}{\ell^{3\ell/2}}, O_\ell(1)}}\cdot O_\ell \pbra{\frac{\ell^{\ell/2}}{n^{\ell/2}}}. \end{equation}
Again, we need to upper bound $\sum_{z\in\{-1,1\}^c}\alpha_z\cdot \ln(1/\alpha_z)^{3\ell/2}$ where $\alpha_z\in[0,1]$ satisfies $\sum_{z\in\{-1,1\}^c}\alpha_z =1$. In the simple case when each $\alpha_z=2^{-c}$, the quantity we wish to bound is $O_{\ell}(c^{3\ell/2})$. A convexity argument similar to the calculation in~\cite{DBLP:conf/approx/GirishRZ21} implies that this is essentially tight and that $\sum_{z\in\{-1,1\}^c}\alpha_z\cdot \ln(1/\alpha_z)^{3\ell/2}\le O_{\ell}(c^{3\ell/2})$. Substituting this in~\Cref{eq:hybrid_4} and doing a calculation similar to that in the proof of~\Cref{lemma:main_lemma_3} implies that
\[ \|\Delta\|_1 \le O\pbra{\frac{(cq)^{3/2}}{\sqrt{n}}} \]
This implies that $cq\ge \Omega(n^{1/3})$ and hence, $c+q\ge\Omega( n^{1/6})$. This completes the proof of~\Cref{theorem:hybrid_2}.

\label{sec:structure}

\bibliographystyle{alpha}
\bibliography{lipics-v2021-sample-article}

\begin{appendix}
	
	\section{Proofs in \texorpdfstring{\Cref{sec:firsttheoremproof}}{Section 3}}
	\label{sec:appendix_theorem_one}
	\begin{proof}[Proof of \Cref{theoremRtwo} using \Cref{lemma:main_lemma}]  The quantum upper bound is presented in~\cite[Theorem 3.8]{DBLP:conf/approx/GirishRZ21}. We describe it below for completeness.
		Let $t=\Theta( k^5\log^3 n \log k)$. For every $x^{(i)},y^{(i)}\in\bin^n$ given as input to Alice and Bob, Alice sends $\ket{x^{(i)}}=\frac{1}{\sqrt{n}}\sum_{j\in [n]}x^{(i)}_j\ket{j}$, and Bob sends $\ket{y^{(i)}}=\frac{1}{\sqrt{n}}\sum_{j\in[n]}y^{(i)}_j\ket{j}$  to the referee. The referee performs a swap test between $\ket{x^{(i)}}$ and $\Had \ket{y^{(i)}}$  and the bias of the swap test is precisely $\forr(x^{(i)},y^{(i)})$, which we are promised is either at least $\eps/2$ or at most $\eps/4$ for every $i\in [k]$ where $\eps=\Theta\pbra{\tfrac{1}{k^2\ln n}}$. The referee takes the threshold of $t=\Theta\pbra{k\cdot \log k \cdot k^4 \ln^2 n}$ many swap tests and a simple calculation similar to~\cite{DBLP:conf/approx/GirishRZ21} shows that the referee can decide $\forr^{(\oplus k)}$  with probability at least $ 2/3$.

		The classical lower bound uses Lemma~\ref{lemma:main_lemma} and~\cite{DBLP:conf/approx/GirishRZ21}. In more detail, \cite{DBLP:conf/approx/GirishRZ21} define two distributions $\tilde{\mu}_{1}^{(k)}$ and $\tilde{\mu}_{-1}^{(k)}$ and prove that these distribution put considerable mass (at least $1-1/\poly(n)$) on the \textsc{yes} and \textsc{no} instances of the $\oplus^k$-Forrelation problem respectively~\cite[Lemma 2.11]{DBLP:conf/approx/GirishRZ21}. They also show~\cite[Theorem 3.1]{DBLP:conf/approx/GirishRZ21} that for any restriction-closed family $\Hcal$ of Boolean functions on $2kn$ variables with outputs in $[-1,1]$, the maximum advantage that functions in $\Hcal$ have in distinguishing $\tilde{\mu}_{1}^{(k)}$ and $\tilde{\mu}_{-1}^{(k)}$ is at most $ O\left( {L_{1,2k}(\mathcal{H})}\cdot  {n}^{-k/2}\right) + o\left(  {n}^{-k/2}\right).$ 
		
		Let $c=\tau\cdot n^{1/4}$ for a small enough constant $\tau>0$. Let $\Hcal$ be the set of all XOR-fibers of $\Rtwoent$ protocols of cost at most $c$ that use $\rho$ as the entangled state. It is not too hard to show that this family is closed under restrictions. Using the aforementioned results, as well as~\Cref{lemma:main_lemma}, we conclude that for all $H\in \Hcal$,
		\[ 
		\left| \underset{ z\sim \tilde{\mu}_{1}^{(k)}}{\E}\left[ H(z) \right]  -  \underset{ z\sim \tilde{\mu}_{-1}^{(k)}}{\E}\left[ H(z) \right] \right| \le O_k\left( 2^{5d}\cdot c^{2k}\cdot n^{-k/2}\right) + o( {n}^{-k/2}) \le O_k(2^{5k}\cdot \tau^{2k}). 
		\] 
		
		In the last step, we used the fact that $d\le k$ and $c=\tau\cdot n^{1/4}$. Setting $\tau\ll 1$ to be a sufficiently small constant, the R.H.S. of the above equation is at most $1/5$. This completes the proof. 
	\end{proof}

	
	
	\section{Proofs in \texorpdfstring{\Cref{sec:r2protocolsstructure}}{Section 3}}
	\label{app:claim2.2proof}
	We prove \Cref{claim:structure_rtwoent} in this section. We begin by describing the structure of $\Rtwoent$ protocols that share a state $\rho\in \Scal(\Hcal_A\otimes \Hcal_B)$ with communication cost $c$. Without loss of generality, the protocol can be described as follows. Alice and Bob could each have private memory consisting of $m$ qubits. Say $\Hcal'_A$ and $\Hcal'_B$ are Hilbert spaces of dimension $2^m$. Note that the dimension $m$ could potentially be very large. Consider a $k$-round protocol. Suppose Alice and Bob got $x,y$ respectively. Below we let $\ket{\Phi_0}=\ket{\Phi}\otimes \ket{0^m}\bra{0^m}_A\otimes \ket{0^m}\bra{0^m}_B$, $z_0=\emptyset$ and $t=0$. Here, the subscript $A,B$ on a qubit denotes which player has that qubit. 
	
	Consider the $(t+1)$th round of the protocol. Suppose Alice and Bob had exchanged messages $z_1,\ldots,z_{2t}\in \bin$ in the first $t$ rounds. Without loss of generality, we can assume that each player sends at most one bit in each round. This assumption can increase the communication cost by a factor of at most two. Alice first applies a two-outcome POVM $\{M_{z_{2t+1}}^{t+1}(x,z_1,\ldots,z_{2t}): z_{2t+1}\in\bin\}$ on the registers that she owns (i.e., her part of the shared state as well her memory). Alice sends the  POVM outcome $z_{2t+1}\in \bin$ to Bob. Based on this message $z_{2t+1}$, his input $y$ and the transcript $z_1,\ldots,z_{2t}\in\bin$, Bob applies  applies a two-outcome POVM $\{N_{z_{2t+2}}^{t+1}(y,z_1,\ldots,z_{2t+1}):z_{2t+2}\in\bin\}$ on the qubits that he owns (his part of the shared state and his memory). He sends the POVM outcome $z_{2t+2}\in\bin$ to Alice. Let the resulting state of all the qubits after the $(t+1)$th round be $\ket{\Phi_{t+1}}$. They repeat this for $k$ rounds after which Alice evaluates some predicate $A$ on $z$ and returns the answer as the output. We say that a protocol has cost $c$, if the size of the transcript is at most $c$~bits. We can assume that there are $\lceil \tfrac{c}{2}\rceil$ rounds and that the players in fact communicate for exactly $\lceil \tfrac{c}{2}\rceil$ rounds, where $c$ is the communication complexity of the protocol on the worst case inputs.

	\begin{proof}[Proof of \Cref{claim:structure_rtwoent}]
		Let $A\subseteq\bin^n$ denote the set of $z\in\bin^n$ that satisfy the final predicate and let $k=\lceil \tfrac{c}{2}\rceil$. Let $z\in \bin^c$. Since $x,y\in \bin^n$ are fixed throughout this proof, for simplicity of notation, for any $j\in[k]$, let
		\[ M^j_z:=M^j_{z_{2j-1}}(x,z_1,\ldots,z_{2j-2})  \quad\text{and}\quad N^j_z:=N^j_{z_{2j}}(x,z_1,\ldots,z_{2j-1}),\]
		\[M^{\le j}_z:=\prod_{j'=j}^1 M^{j'}_z \quad\text{ and } N^{\le j}_z:=\prod_{j'=j}^1 N^{j'}_z.\]
		Let $M_z=M_z^{\le k}$ and $N_z=N_z^{\le k}$. Note that the $M_z^j$ are functions of $x,z$ and $N_z^j$ are functions of $y,z$. Finally, let $E_z(x):={M_z}^\dagger M_z$ and  $F_z(y):={N_z}^\dagger N_z$. It is clear that properties 1 and 2 are satisfied by $E_z$ and $F_z$. It is straightforward to see from the definition of $\Rtwoent$ protocols that the expected output of the protocol is precisely
		\begin{align*}
			&C(x,y)\\
			&=(-1)\cdot \sum_{z\in A}\Tr\pbra{\pbra{M_z \otimes N_z } \rho'  \pbra{M^\dagger_z \otimes N^\dagger_z }}  + 1\cdot \sum_{z\notin A}\Tr\pbra{\pbra{M_z \otimes N_z } \rho'  \pbra{M^\dagger_z \otimes N^\dagger_z }}  \\
			&=\sum_{z\in \bin^c}\Tr\pbra{\pbra{E_z(x)\otimes F_z(y)} \rho'  }\cdot (-1)^{\indi[z\in A]}. 
		\end{align*}
		It only remains to prove property 3. Consider:
		\begin{align*}
			&(*):=\sum_{z\in\bin^c} E_z(x)\otimes F_z(y)\\
			&\triangleq \sum_{z\in \bin^{2k}}\sbra{ \pbra{M^{\le k}_z}^\dagger\cdot \pbra{M^{\le k}_z}}\otimes \sbra{ \pbra{N_z^{\le k}}^\dagger \cdot N_z^{\le k}}\\
			&= \sum_{z\in \bin^{2k-1}}\sbra{ \pbra{M^{\le k}_z}^\dagger\cdot \pbra{M^{\le k}_z}}\otimes\sbra{ \pbra{N^{\le k-1}_z}^\dagger\cdot \pbra{\sum_{z_{2k}\in\bin} \pbra{N^k_z}^\dagger\cdot N^k_z}\cdot\pbra{N^{\le k-1}_z}
			} \end{align*}
		The last equality used the fact that the operators $N_z^{\le k-1}$ and $M_z^{\le k}$ do not depend on $z_{2k}$. For all $z\in \bin^{2k-1}$, $\cbra{N^{k}_z:z_{2k}\in\bin }$ is a two-outcome POVM, thus, 
		\[ \sum_{z_{2k}\in \bin} \pbra{N^{k}_z}^\dagger\cdot N^{k}_z=\mathbb{I}.\] Let $N=\pbra{N_z^{\le k-1}}^\dagger\cdot N_z^{\le k-1}$. Substituting this above, we have
		\begin{align*}
			&(*)= \sum_{z\in \bin^{2k-1}} \sbra{ \pbra{M_z^{\le k}}^\dagger\cdot M_z^{\le k}}\otimes N \\
			&= \sum_{z\in \bin^{2k-2}} \sbra{\pbra{M^{\le k-1}_z}^\dagger\cdot \pbra{\sum_{z_{2k-1}\in\bin} \pbra{M^k_z}^\dagger\cdot M^k_z}\cdot\pbra{M^{\le k-1}_z}}\otimes N.
		\end{align*}
		The last equality used the fact that the operators $M_z^{\le k-1}$ and $N=\pbra{N_z^{\le k-1}}^\dagger\cdot N_z^{\le k-1}$ don't depend on $z_{2k-1}$. For all $z\in\bin^{2k-2}$,  $\cbra{M^{j}_z:z_{2k-1}\in\bin }$ is a two-outcome POVM, thus, 
		\[ \sum_{z_{2k-1}\in \bin} \pbra{M^{k}_z}^\dagger\cdot M^{k}_z=\mathbb{I}. \] 
		Substituting this above, we have
		\begin{align*}
			(*)&=\sum_{z_1,\ldots,z_{2k-2}\in \bin} \sbra{\pbra{M^{\le k-1}_z}^\dagger\cdot \pbra{M^{\le k-1}_z}}\otimes \sbra{\pbra{N^{\le k-1}_z}^\dagger\cdot\pbra{N^{\le k-1}_z}}
			\\
			&=\ldots=\mathbb{I}\quad  \text{by induction on } k.
		\end{align*}
		This proves property 3 and completes the proof of~\Cref{claim:structure_rtwoent}.
	\end{proof}
	
	\section{Quantum upper bound in Theorem~\ref{theoremRoneQsmp}.}
	\label{sec:appendix_theorem_three}
	
	We discuss the quantum upper bound for a single instance of the Boolean Hidden Matching problem (this is similar to the protocol in~\cite{buhrman2011near}). Here Alice and Bob share the quantum state $\frac{1}{\sqrt{n}}\sum_{i=1}^n\ket{i}_A\ket{i}_B$. Alice applies the transformation $\ket{i}_A\rightarrow (-1)^{x_i}\ket{i}_A$ and transforms the state to 
	$
	\frac{1}{\sqrt{n}}\sum_{i\in [n]} (-1)^{x_i}\ket{i}_A\ket{i}_B.
	$
	Bob now measures his register in the matching basis, in particular, he completes his $\alpha$-partial matching arbitrarily to a complete matching. Let the resulting complete  matching be $\{(e_i,e_j)\}$ wherein $\alpha$-fraction of the edges belong to $\mathcal{E}$. Now Bob measures in the basis $\{\ketbra{e_{i}}{e_{i}}+\ketbra{e_{j}}{e_{j}}\}$. Now, Bob obtains a uniformly random edge in the matching (known to him). Furthermore, with probability $1/\alpha$, the obtained edge was in $\mathcal{E}$. Say he obtained $(e_i,e_j)\in \mathcal{E}$. The state then collapses to
	\begin{align}
		\frac{1}{\sqrt{2}}\Big((-1)^{x_{i}}\ket{i}\ket{i}+(-1)^{x_{j}}\ket{j}\ket{j}\Big).
	\end{align}
	Note that Bob knows the edge $(i,j)$. Now, both Alice and Bob apply the $(\log n)$-qubit Hadamard gate on their respective registers, the resulting state is given by
	$$
	\frac{1}{\sqrt{2}n}\sum_{a,b\in \{0,1\}^{\log n}}\Big((-1)^{x_{i}+(a+b)\cdot i}+(-1)^{x_{j}+(a+b)\cdot j}\Big)\ket{a,b}.
	$$
	Now observe that  if Alice and Bob measure their respective registers, Alice obtains $a$ uniformly random,  Bob obtains $b$ satisfying $(i\oplus j)\cdot (a\oplus b)=x_i+x_j$. Alice sends $a$ and Bob sends $(i,j)$, $y_{ij}$ as well as $b$ to the referee. The referee now can now compute $(i\oplus j)\cdot (a\oplus b)$ and learn $x_i\oplus x_j$.  The referee returns \textsc{no} if $y_{ij}=x_i\oplus x_j$, and returns \textsc{yes} if $\overline{y}_{ij}=x_i\oplus x_j$. This solves the Boolean Hidden Matching Problem.  Observe that this protocol succeeds with probability $\alpha$ (which is the probability that Bob's measurement gives an edge in his matching input). 
	
	In order to compute $\bhm^{(\oplus k)}$ Alice and Bob perform the following: for every $i\in [k]$, they carry out the protocol $O((\log k)/\alpha)$ many times and send all their outcomes to the referee. With probability at least $ 9/10$, the measurement collapses to an edge in the matching, which Bob knows and can communicate to the referee. For this edge, the referee checks the predicate if $y_{ij}=x_i\oplus x_j$ is satisfied or not and hence knows the value of $\bhm(x^i,y^i)$. Hence after $O((\log k)/\alpha)$ bits of communication, the referee knows 
	$\bhm(x^i,y^i)$ for \emph{all} $i\in [k]$ and hence~$\bhm^{(\oplus k)}$.
	\subsection{Proofs in \texorpdfstring{\Cref{sec:five_lower_bound}}{Section 5}}


	\begin{proof}[Proof of~\Cref{fact:matchingprobability}] Let $|S_i|=2\ell_i$ for $i\in[k]$ such that $\sum_{i\in[k]}\ell_i=\ell$.
		It is argued in~\cite{gavinsky2007exponential} that the probability that a random matching on $[n]$ of size $m=\alpha n$ matches any given set of size $2\ell_i$ is precisely $\frac{{\alpha n\choose \ell_i}}{{n\choose 2\ell_i}}$.
		Furthermore they showed that $\frac{{\alpha n \choose \ell_i}}{{n\choose 2\ell_i}}$ is a decreasing function of $\ell_i$, and is at most $ O_{\ell_i}\pbra{(\ell_i/n)^{\ell_i}} $. Thus, the probability that $M$ matches $S$~is
		\begin{align*}g(\ell_1,\ldots,\ell_k)=\prod_{i\in[k]} \frac{{\alpha n \choose \ell_i}}{{n\choose 2\ell_i}}&\le \pbra{\max_{i\in[k]}\frac{{\alpha n \choose \ell_i}}{{n\choose 2\ell_i}}}^k\le  \pbra{ \frac{{\alpha n \choose \ell/k}}{{n\choose 2\ell/k}}}^k = O_{\ell}\pbra{\frac{\ell^{\ell}}{(nk)^{\ell}}}.
		\end{align*}
	\end{proof}
	

\end{appendix}
\end{document}